\documentclass[pra, 12pt, showpacs,floatfix,nofootinbib]{revtex4-1} %
\usepackage{amssymb,amsmath,amsfonts,amsthm}
\usepackage[utf8]{inputenc}
\usepackage{graphicx}
\usepackage{verbatim}
\usepackage{array}
\usepackage{tensor}
\usepackage{hyperref}
\usepackage{epsfig}
\newcommand{\bydef}{\stackrel{\mathrm{def}}{=}}

\usepackage{tweaklist}
\renewcommand{\enumhook}{\setlength{\topsep}{-0.5pt}%
  \setlength{\itemsep}{-0.5pt}}

\hypersetup{
colorlinks=false, linkcolor=black, urlcolor=black, pdfpagemode=none,
pdftitle={Categorical Tensor Network States}, pdfauthor={Biamonte}, pdfcreator={$
$Id:categorical-tensor-networks.tex, \today$ $},
  pdfsubject={Categorical Tensor Network States},
  pdfkeywords={Tensor Network States, category theory, Bialgebra, Hopf
Algebra, Boolean Algebra, Quantum Theory, Entanglement, Map State Duality, MERA, PEPS, MPS, Quantum Circuits}
}

\newcommand{\bra}[1]{\mbox{$\langle #1|$}}
\newcommand{\ket}[1]{\mbox{$|#1\rangle$}}
\newcommand{\braket}[2]{\mbox{$\langle #1|#2\rangle$}}
\newcommand{\ketbra}[2]{\mbox{$|#1\rangle\langle #2|$}}

\def\W{{\sf W}}
\newcommand{\gate}[1]{\ensuremath{\text{\sf #1}}}

\newcommand{\CNOT}{\gate{CNOT}}
\newcommand{\NOT}{\gate{NOT}}

\newcommand{\XOR}{{\sf XOR}}
\newcommand{\OR}{{\sf OR}}
\newcommand{\COPY}{{\sf COPY}}
\newcommand{\AND}{{\sf AND}}
\newcommand{\NAND}{{\sf NAND}}
\newcommand{\XNOR}{{\sf XNOR}}
\newcommand{\NOR}{{\sf NOR}}
\newcommand{\ANF}{{\sf ANF}}
\newcommand{\PPRM}{{\sf PPRM}}

\def\NOT{{\sf NOT}}

\def\H{{\sf H}}

\newcommand{\GHZ}{{\sf GHZ }}

\newcommand{\eye}{\mathbf{1}} 

\newcommand{\be}{\begin{equation}}
\newcommand{\ee}{\end{equation}}

\newcommand{\hilbert}[1]{\ensuremath{\mathcal{#1}}}

\theoremstyle{plain}
\newtheorem{theorem}{Theorem}
\newtheorem{lemma}[theorem]{Lemma}
\newtheorem{corollary}[theorem]{Corollary}
\newtheorem{remark}[theorem]{Remark}
\newtheorem{observation}[theorem]{Observation}
\theoremstyle{definition}
\newtheorem{definition}[theorem]{Definition}
\newtheorem{example}[theorem]{Example}

\def\1#1{{\bf #1}}
\def\2#1{{\cal #1}}
\def\3#1{{\sl #1}}
\def\4#1{{\tt #1}}
\def\5#1{{\sf #1}}
\def\6#1{{\mathfrak #1}}
\def\7#1{{\mathbb #1}}

\begin{document}

\title{Categorical Tensor Network States}
\author{Jacob D. Biamonte,$^{1,2,3,}$\footnote{\tt jacob.biamonte@qubit.org} Stephen R.
Clark$^{2,4}$ and Dieter Jaksch$^{5,4,2}$}
\affiliation{$^1$Oxford University Computing Laboratory, Parks Road
Oxford, OX1 3QD, United Kingdom\\
$^2$Centre for Quantum Technologies, National University of Singapore, 3 Science
Drive 2, Singapore 117543, Singapore \\
$^3$ISI Foundation, Torino Italy \\ 
$^4$Keble College, Parks Road, University of Oxford, Oxford OX1 3PG, United Kingdom\\
$^5$Clarendon Laboratory, Department of Physics, University of Oxford, Oxford OX1
3PU, United Kingdom}

\pacs{03.65.Fd, 03.65.Ca, 03.65.Aa}

\begin{abstract}
We examine the use of string diagrams and the mathematics of category theory in the
description of quantum states by tensor networks.  This approach lead to a unification of several 
ideas, as well as several results and methods that have not previously appeared in either side of the literature.  
Our approach enabled the development of a tensor network framework 
allowing a solution to the 
quantum decomposition problem which has several appealing features. Specifically, given an $n$-body quantum state 
$\ket{\psi}$, we present a new and general method to factor $\ket{\psi}$ into a tensor
network of clearly defined building blocks.  We use the solution to expose a previously unknown and large class of quantum states which we prove can be sampled
efficiently and exactly. This general framework of categorical tensor network
states, where a combination of generic and algebraically defined tensors appear,
enhances the theory of tensor network states.
\end{abstract}
\maketitle

\section{Introduction}

Tensor network states have recently emerged from Quantum Information
Science as a general method to simulate quantum systems using classical computers.  
By utilizing quantum
information concepts such as entanglement and condensed matter concepts like
renormalization,
several novel algorithms, based on tensor network states (TNS), have been developed
which have overcome many pre-existing limitations. These and other related methods
have been used to perform highly accurate calculations on a broad class of
strongly-correlated systems and have attracted significant interest from several
research communities concerned with computer simulations of physical systems.

In this work we develop a tool set and corresponding framework
to enhance the range of methods currently used to address
problems in many-body physics.  In this categorical network
model of quantum states, each of the internal components that form the
building blocks of the network  are defined in terms of their mathematical
properties, and these properties are given in terms of equations which have a graphical interpretation. In this way, diagrammatic methods from modern algebra and category theory~\cite{MacLane} can be combined with graphical methods currently used in tensor network descriptions of many-body physics.  Moreover our results indicate that it may be 
be advantageous to use ``categorical components'' within tensor networks, whose
algebraic properties can permit a broader means of rewiring networks, and
potentially reveal new types of contractible tensor networks.  Our results
include defining a new graphical calculus on tensors, and exposing their key
properties in a ``tensor tool box''. We use these tensors to present one solution to
the problem of factoring any given quantum state into a tensor network. The graphical properties of this solution
then enabled us to expose a wide class of tensor network states that can be sampled
in the computational basis efficiently and exactly.  Here we list several novel contributions of the present paper.  

\begin{enumerate}
 \item[(i)] We introduce a universal class of tensors to the quantum theory and tensor networks literature.  We cast several properties and simplification rules applicable for classical networks of this type, into the language of tensor networks.  (This fixed collection of tensors and the corresponding Boolean algebra framework we introduce provides the potential for a new tool to construct networks of relevance for problems in quantum information science and condensed matter physics. Recent work has used the collection of tensors we introduce here (arXiv version) together with their rewrite rules to solve models in lattice gauge theory \cite{2011arXiv1108.0888D} and to study correlator product states \cite{2011arXiv1107.0936A}).   
  \item[(ii)] We define a new class of quantum states, quantum Boolean states, expressed as: $\ket{\psi} = \sum f(x) \ket{x}$ where $f$ is a switching function and the sum is over all n-long bit strings.   A quantum state is called Boolean iff it can be written in a local basis with amplitude coefficients taking only binary values $0$ or $1$.  Examples of states in this class include states such as \W{} and \GHZ{}-states.  We develop tools which aid in the study of this class of states. 
 \item[(iii)] We prove that a tensor network representing a Boolean quantum state is determined from the classical network description of the corresponding function.  The quantum tensor networks are found by the following method:  we let each classical gate act on a linear space and replace the composition of functions, with the contraction of tensors.  This technique is detailed in the present work and has very recently been used in \cite{2011arXiv1108.0888D}.     
 \item[(iv)] We present a proof which shows a new universal constructive decomposition of any quantum state, into a network built in terms of tensors in the tool box we introduce.  
 \item[(v)] We present a new and large class of quantum states which we prove to be efficient to sample.  A subclass of this class of states includes the widely studied class of correlator product states.  This connection (arXiv version) was recently studied in \cite{2011arXiv1107.0936A}.  Exploring this class in further detail is left to future studies.  
  \item[(vi)] Although Boolean algebra has long been used in both classical and quantum computer science, we seem to be the first to 
 tailor its use to the problem of describing quantum states by tensor networks.  This has resulted in new proof techniques and also enables one to use methods from the well developed graphical language of classical circuits inside the domain of tensor networks.   
\end{enumerate}
   
To explain the main motivation which prompted us to study tensor network states, let us recall
the success of established numerical simulation methods, such as the density matrix
renormalization group (DMRG)~\cite{DMRG,RevModPhys.77.259, 2009PhRvL.102e7202H} which is based
on an elegant class of tensor networks called Matrix Product States (MPS)~\cite{PhysRevLett.75.3537}. For more
than 15 years DMRG has been a key method for studying the stationary properties of
strongly-correlated 1D quantum systems 
in regimes far beyond those which can be described with perturbative or mean-field
techniques.  Exploiting the tensor network structure of MPS has lead to explicit algorithms, such as
the Time-Evolving Block Decimation (TEBD)
method~\cite{vidal-2003-91, vidal-2004-93,
2009PhRvA..80d2333B,2010PhRvB..82k5126P,2010PhRvB..82l5118K,course}, for computing
the real-time dynamics of 1D quantum systems. Accurate calculations of
out-of-equilibrium properties has proven extremely useful for describing various
condensed matter systems~\cite{2005PhRvA..72d3618D, 2010arXiv1010.2351S,
2010NJPh...12b5005C}, as well as transport phenomena in ultra-cold atoms in optical
lattices~\cite{PhysRevA.76.011605,PhysRevA.82.043617,1367-2630-12-2-025014}.
Additionally the TEBD method has recently been successfully adapted to the simulation
of stochastic classical systems~\cite{tommy}, as well as for simulating operators in
the
Heisenberg picture~\cite{2009PhRvL.102e7202H}. Despite these successes,
limitations remain in the size, dimensionality, and classes of Hamiltonians that can
be simulated with MPS based methods. To overcome these restrictions, several new 
algorithms have been proposed which are based on different types of tensor network states.
Specifically: Projected Entangled Pair States
(PEPS)~\cite{MPSreview08,2010PhRvA..81e2338K,2010PhRvB..81p5104C} which directly generalize the MPS
structure to higher dimensions, and the Multi-scale Entanglement Renormalization Ansatz
(MERA)~\cite{MERA,2009arXiv0912.1651V,2009PhRvB..80p5129C,2010PhRvA..81a0303C}) which
instead utilizes an intuitive hierarchical structure.

Category theory is often used as a unifying language for
mathematics~\cite{MacLane} and in more recent times has been
used to formulate physical
theories~\cite{prehistory}.  One of the
strong points of categorical modeling is that it comes equipped with many types of
intuitive graphical calculi. (We mention the coherence
results~\cite{coherence, MacLane, Selinger09} and as a matter of
convenience, make use of $\dagger$-compactness. The
graphical calculus of categories
formally extends to a rigorous
tool.  See for instance, Selinger's survey of graphical languages for monoidal
categories outlining the categories describing quantum
theory~\cite{Selinger09}.)  The graphical theory behind the types of diagrams we consider here dates to the 
work of Lafont \cite{Lafont92, Lafont95} who built directly on the earlier work by Penrose \cite{Penrose}. 

A motivating reason for connecting category theory to tensor networks is that, increasingly,
both existing and newly developed tensor network algorithms are most easily expressed
in terms of informal graphical
depictions. This graphical approach can now be complemented and enhanced by
exploiting the long existing 
language of category theory~\cite{MacLane,2009arXiv0903.0340B,prehistory}. This
immediately enables the application
of many established techniques allowing for both a ``zoomed out'' description exposing
known high-level structures, but also enables new ``zoomed in'' descriptions, exposing
``hidden'' algebraic structures that have not previously been considered. 

We will illustrate our categorical approach to tensor networks by focusing on tensor
networks constructed from familiar components, namely Boolean logic gates (and
multi-valued logic gates in the case of qudits), applied to this
unfamiliar context. To accomplish this goal,
we build on ideas across several fields. This includes extending the work
by Lafont~\cite{Boolean03} which was aimed at providing an algebraic
theory for classical circuits. (Lafont's work is related to the more recent
work on proof theory by Guiraud~\cite{3Dproofs}, and is a different direction
from other work on applying category theory to classical networks appearing
in~\cite{CatCircuits}.) The use of symmetric monoidal categories tightens this
approach and removes some redundancy in Lafont's graphical
lemmas. The application of these results to tensor networks
introduces several
novel features. The first feature is that once conventional logic circuits are
formulated as tensor networks they can be distorted into atemporal configurations
since the indices (or legs) of tensors can be bent around arbitrarily. This permits
a very compact tensor network representation of a large interesting class of Boolean
states such as {\sf GHZ}-states, {\sf W}-states and 
symmetric states~\cite{2010NJPh...12g3025A}, using exclusively Boolean gate
tensors. 
A second feature is that once expressed as tensors the corresponding classical logic
circuits act on
complex valued inputs and outputs, as opposed to just binary values. By permitting
arbitrary single-qubit states (general rank-1 tensors) at the output of tensor
networks, which are otherwise composed of only switching functions, we arrive at a
broader class of generalized Boolean states. We prove in Theorem \ref{theorem:rep}
that this class of states provides an explicit construction method for factoring any
given quantum state into a 
tensor network. As expected, the cost of this exhaustiveness is that the resulting
network is, in general, neither efficient in description or in contraction.  However, by
limiting both the gate count and number of the switching functions
comprising the tensor networks to be polynomial in the system size, we obtain a
class of states, which we call Generalized Polynomial Boolean States
(GPBS see Definition \ref{def:GPBS}), that can be sampled
in the computational basis efficiently and exactly (Theorem
\ref{corollary:states}).

\paragraph{Manuscript Structure.}
Next in Section \ref{sec:overview} we quickly review the key concepts
introduced in this paper before going into detail in the remaining sections.  We
continue in Section \ref{sec:components} by defining the network building blocks:
this
includes defining some new rank-3 tensors such as the quantum \AND-state in
Equation~\eqref{eqn:quantumAND}. We then consider how these components interact in
Section~\ref{sec:interaction}.  This is done in terms of algebraic structures, such
as
Bialgebras (Section \ref{sec:bialgebra}) and Hopf-algebras (Section \ref{sec:hopf})
which are well known to have a purely diagrammatic interpretation.  With these definitions in place,
in Section~\ref{sec:CTN} we apply this framework to tensor network theory. 
As an illustrative example, we consider the \W- and {\sf GHZ}-states using our formalism.  Specifically, we consider a
particular categorical tensor
network for many-body \W-states in Section \ref{sec:CTN}. A proof of our
decomposition
theorem for quantum states is given in Section~\ref{sec:mainresult}.  In conclusion,
we mention some
future directions for work in Section \ref{sec:conclusion}. We have included
Appendix~\ref{sec:newalgebra} on algebras defined on quantum states and
Appendix~\ref{sec:Boolean} on the Boolean \XOR-algebra.

\paragraph{Background Reading.} 
The results appearing in this work were found by tailoring several powerful
techniques from modern mathematics: category theory, algebra and co-algebra and
applicable results from classical network theory and graphical calculus. Tensor
network states are covered in the reviews~\cite{MPSreview08,TNSreview09,
2009arXiv0912.1651V,2010arXiv1001.0767S}. For general background on category theory
see~\cite{MacLane}.  For background on Boolean algebra, discrete set
functions and circuit theory see~\cite{Weg87} and see \cite{BH02} for background on
pseudo Boolean functions and for multi-valued logic see~\cite{MVL1}.  For background
on quantum circuits and quantum computing concepts see~\cite{KSV02, NC00} and for
background on the theory of entanglement see~\cite{2008RvMP...80..517A}.  For the
current capabilities of the existing graphical language of tensor network states see
e.g.\ \cite{Atemp06,GESP07} and for work on using ideas related to tensor networks
for state preparation of physical systems
see~\cite{2008PhRvL.101r0506L,2007PhRvA..75c2311S,GESP07}.

\section{Results Overview}\label{sec:overview}

In the present Section, we informally review our main results. 
The idea of translating any given quantum state or 
operator into a representation in terms of a connected network of algebraically defined components is explained next in
Section \ref{sec:newrep} with the corresponding algebraic definitions of these
network components over viewed in Section \ref{sec:comdef}.  Boolean quantum
logic tensors are then introduced in Section \ref{sec:c}.  We summarize our main
results in Section \ref{sec:connectingdots}. 

\subsection{Tensor network representations of quantum states}\label{sec:newrep}
A qudit is a $d$-level generalization of a qubit.  In
physics a quantum state of $n$-qudits has an exact representation as a rank-$n$
tensor with each of the open legs corresponding to a physical degree of freedom, such
as a spin with $(d-1)/2$ energy levels. Such a representation, shown in
Figure~\ref{fig:tensor-networks}(a) is manifestly inefficient since it will have a
number of complex components which grows exponentially with $n$. The purpose of
tensor network states is to decompose this type of structureless rank-$n$ tensor into
a network of tensors whose rank is bounded. There are now a number of ways to
describe strongly-correlated quantum lattice systems as 
tensor-networks. As mentioned in the introduction, these include
MPS~\cite{ostromm1995,fannesnachtwern1992,2010NJPh...12b5005C},
PEPS~\cite{verstraete08,TNSreview09} and MERA~\cite{vidal2007, 2009arXiv0912.1651V}.
For MPS and
PEPS, shown in
Figures~\ref{fig:tensor-networks}(b) and (c), the resulting network of tensors
follows
the geometry of the underlying physical system, e.g., a 1D chain and 2D grid,
respectively. Alternatively a Tensor Tree Network
(TTN)~\cite{PhysRevA.74.022320,2009PhRvB..80w5127T} can be
employed which has a hierarchical structure where only the bottom layer has open physical legs, as shown in Figure~\ref{fig:tensor-networks}(d) 
for a 1D system and Figure~\ref{fig:tensor-networks}(e) for a 2D one. (Each
tensor in these networks is otherwise
unconstrained, although enforcing some constraints, such as orthogonality, has
numerical advantages.) For MERA the network is similar
to a TTN, as seen in Figure~\ref{fig:tensor-networks}(f) for 1D, but is instead comprised of alternating layers of rank-4 unitary and rank-3 isometric
tensors. The central problem faced by all types of tensor networks is that the
resulting tensor network for the quantity $\bra{\psi}(\2 O\ket{\psi})$, where $\2O$
is some product operator, needs to be efficiently contractible if any
physical results, e.g., expectation values, correlations or probabilities, are to be
computed. For MPS and TTN
efficient exact contractibility follows from the 1D chain or tree-like geometry, while for MERA it follows from its peculiar causal cone
structure resulting from the constraints imposed on the tensors~\cite{2009arXiv0912.1651V}. For
PEPS, however, exact contraction is not efficient in general, but  can be rendered
efficient if
approximations are made~\cite{verstraete08,TNSreview09}.

\begin{figure}[t]
\includegraphics[width=0.9\textwidth]{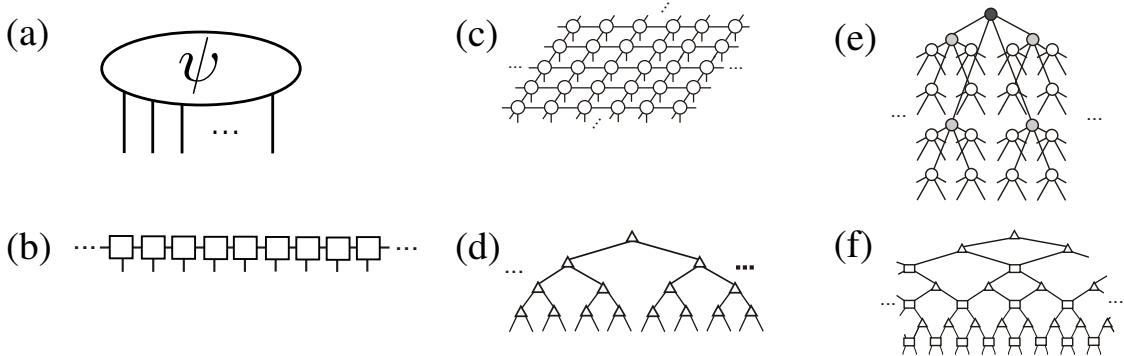}
\caption{(a) A generic quantum state $\ket{\psi}$ for $n$ degrees of freedom
represented as a tensor with $n$ open legs. (b) A comb-like MPS tensor network for a
1D chain system~\cite{ostromm1995,fannesnachtwern1992}. (b) A grid-like PEPS tensor
network for a 2D lattice system~\cite{verstraete08,TNSreview09}. (d) A TTN for a 1D
chain system where only the bottom layer of tensors possess open physical
legs~\cite{PhysRevA.74.022320,2009PhRvB..80w5127T}. (e) A TTN for a 2D lattice
system. (f) A hierarchically structured MERA network for a 1D chain system possessing
unitaries (rank-4 tensors) and isometries (rank-3 tensors)~\cite{vidal2007,
2009arXiv0912.1651V}. This tensor network can also be generalized to a 2D lattice
(not shown).}\label{fig:tensor-networks}
\end{figure}

In our approach we define a categorical tensor network state (CTNS) generally as any
TNS which contains some algebraically constrained tensors along with possible generic
ones. Indeed, when recast, certain widely used classes of TNS can be readily exposed
as examples of CTNS.
Specifically, variants of PEPS have been proposed called
string-bond states~\cite{PhysRevLett.100.040501}.  Although these string-bond states,
like PEPS in general, are not
efficiently contractible, they are efficient to sample. By this we mean that for
these
special cases of PEPS any given amplitude of the resulting state (for a fixed computational basis state) can be extracted exactly and efficiently,
in contrast to generic PEPS. This permits variational quantum Monte-Carlo calculations to be performed on string-bond states 
where the energy of the state is stochastically
minimized~\cite{PhysRevLett.100.040501}. This remarkable property
follows directly from the use of a tensor, 
called the {\sf COPY}-dot, which will form one of several tensors in the fixed
toolbox considered in great detail later. As its name suggests, the {\sf COPY}-dot
duplicates inputs states in the computational basis, and thus with these inputs
breaks up into disconnected components, as depicted in
Figure~\ref{fig:stringbonds}(a). By using the
{\sf COPY}-dot as the ``glue" for connecting up a TNS, the ability to sample the
state
efficiently is guaranteed so long as the individual parts connected are themselves
contractible. The generality and applicability of this trick can be seen by examining
the structure of string-bond states, as well as other types of 
similar states like entangled-plaquette-states~\cite{2009NJPh...11h3026M} and
correlator-product states~\cite{2009PhRvB..80x5116C}, shown in Figure~\ref{fig:stringbonds}(c)-(e).  A long-term aim of this work is that by presenting our toolbox of
tensors, entirely new classes of CTNS with similarly desirable contractibility
properties can be devised. Indeed we have a useful result in this direction
(Theorem \ref{corollary:states}) by
introducing a new class of states, which can also be sampled exactly and efficiently.

\begin{figure}[t]
\includegraphics[width=0.9\textwidth]{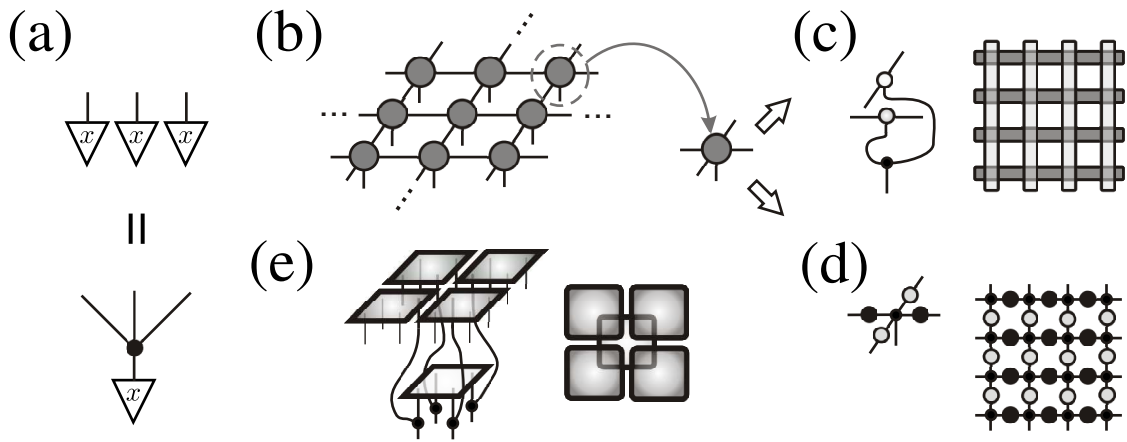}
\caption{(a) One of the simplest tensors, called the diagonal in category theory,
the {\sf COPY}-gate or the {\sf COPY}-dot in circuits, copies computational basis
states
$\ket{x}$ where $x=0,1$ for qubits and $x=0,1,...,d-1$
for qudits.  The tensor subsequently breaks up into disconnected states.
(b) A generic PEPS in which we expose a single generic rank-5 tensor. This
tensor network can neither be contracted nor sampled exactly and efficiently.
However,
if the tensor has internal structure exploiting the {\sf COPY}-dot then efficient
sampling
becomes possible. (c) The tensor breaks up into a vertical and a horizontal rank-3
tensor joined by the {\sf COPY}-dot. Upon sampling computational basis states the
resulting contraction reduces to many isolated MPS, each of which are exactly
contractible, for
each row and column of the lattice. This type of state is known as a string-bond
state and can be readily generalized~\cite{PhysRevLett.100.040501}. (d) An even
simpler case is to break the tensor
up into four rank-2 tensors joined by a {\sf COPY}-dot forming a co-called
correlator-product state~\cite{2009PhRvB..80x5116C}. (e) Finally, outside the PEPS
class, there are entangled
plaquette states~\cite{2009NJPh...11h3026M} which join up overlapping tensors (in
this case rank-4 ones
describing a $2\times 2$ plaquette) for each plaquette. Efficient sampling is again
possible due to the {\sf COPY}-dot.}\label{fig:stringbonds}
\end{figure}

On an interesting historical note, to the best of our knowledge, a graphical
interpretation of tensors was first pointed out in~\cite{Penrose}.  A graphical language for describing the manipulations and
steps of tensor network based algorithms has become widely
used. By introducing new tensors, and by considering their graphical properties, an aim of our work is to extend the existing methods of the diagrammatic methods used. 
 The graphical properties of tensors are
defined succinctly via so-called \textit{string diagrams}. As an exemplary illustration, we
will
consider
in great detail CTNS which are composed entirely from a tensor toolbox built 
from classical Boolean logic gates. By invoking known theorems asserting the
universality of multi-valued logic~\cite{MVL1} (also called $d$-state switching), our
methods can be
readily applied to tensors of any finite dimension.  Our approach provides not only
an
example of the use of
well known gates in an unfamiliar
context but also illustrates the potential power of having ``algebraic components''
within a tensor network. The next two
sections highlight some of the properties of this toolbox of tensors and
reviews our main result showing a new quantum state decomposition using a subset of
them. The full details of this work then follow from Section~\ref{sec:components}
onwards.

\subsection{Network components fully defined by diagrammatic laws}\label{sec:comdef}

We will now review the set of tensors that form our universal
building blocks.  To get an idea of how
the tensor calculus will work, consider Figure~\ref{fig:F2-presentation},
which forms a presentation of the linear fragment of the Boolean
calculus~\cite{Boolean03}):
that is, the calculus of Boolean algebra we represent on quantum states, restricted
to the building blocks that can be used to generate linear Boolean functions.  

To recover the full Boolean-calculus, we must consider a non-linear Boolean gate: we
use the \AND{}-gate.  Figure~\ref{fig:F2-presentation} together with
Figure~\ref{fig:extraF2} form a full presentation of the
calculus~\cite{Boolean03}. The origin and consequences of these relations will be
considered in full detail in Section~\ref{sec:components}. The presentations in Figure~\ref{fig:F2-presentation} together with
Figure~\ref{fig:extraF2} represent a
complete set of defining equations~\cite{Blog}. The results we report and the
introduction of this new picture calculus into physics has already attracted
significant interest and provided a new research direction in categorical quantum
mechanics~\cite{catQM,Selinger2007139}. 

\begin{figure}[h]
\includegraphics[width=0.9\textwidth]{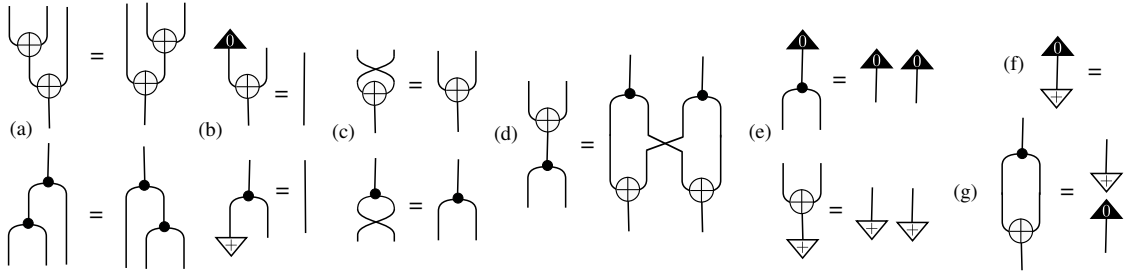}
\caption{Read top to bottom. A presentation of the linear fragment of the Boolean
calculus. The plus
($\oplus$) dots are \XOR{} and the black ($\bullet$) dots represent
{\sf COPY}.  The details of (a)-(g) will be given in
Sections~\ref{sec:components} and \ref{sec:interaction}.  For
instance, (d) represents the bialgebra law and (g) the Hopf-law (in
the case of qubits $x\oplus x =0$, in higher dimensions
the units $\bra{+}$ becomes
$\bra{0}+\bra{1}+\cdots+\bra{d-1}$).}\label{fig:F2-presentation}
\end{figure}

\begin{figure}[h]
\includegraphics[width=0.9\textwidth]{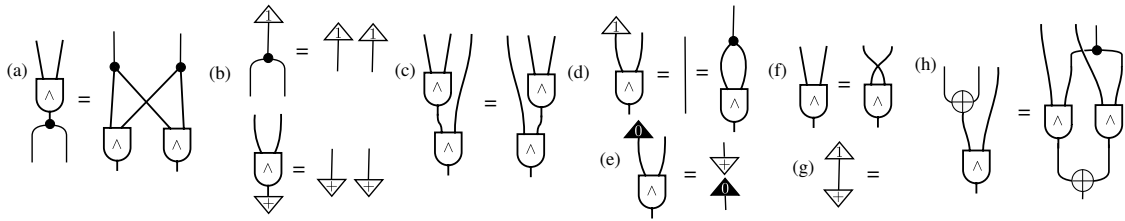}
\caption{Diagrams read top to bottom.  A presentation of the Boolean-calculus with
Figure~\ref{fig:F2-presentation}. The details of (a)-(g) will be given in
Sections~\ref{sec:components} and \ref{sec:interaction}.  For
instance, (h) represents distributivity of \AND{}($\wedge$) over \XOR{} ($\oplus$),
and (d) shows that $x\wedge x = x$.}\label{fig:extraF2}
\end{figure}

Proceeding axiomatically we need to add a bit more to the presentation of the
Boolean calculus to represent operators and quantum states. This is because e.g.\ all
the diagrams in Figure~\ref{fig:F2-presentation} and \ref{fig:extraF2} are
read from the top of the page to the bottom. Our network model of quantum states
requires that we are
able to turn maps upside down, e.g.\ transposition.  This additional flexibility
comes from an added ability to bend wires.  We can hence define transposition
graphically (see Figure~\ref{fig:adjoints} (d)). 

The way forward is to add what category theory refers to as \textit{compact
structures} (see
Section~\ref{sec:bends} for
details).  These compact structures are given diagrammatically as
\begin{center}
\includegraphics[width=0.3\textwidth]{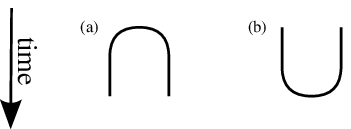}
\end{center}
and as will be explored in Section~\ref{sec:bends} these two structures allow us to
formally bend wires and to define the transpose of a linear
map/state, and provide a formal way to reshape a matrix.  We understand (a) above as
a cup, given as the generalized
Bell-state $\sum_{i=0}^{d-1}\ket{ii}$ and (b) above as the so-called cap,
Bell-costate
$\sum_{i=0}^{d-1}\bra{ii}$ or \textit{effect}. (Normalization factors
omitted: 
without loss of generality, we will often omit global scale factors (tensor
networks with no open legs).  This is done for ease of
presentation.  We note that for Hilbert space $\2 H$ there is a natural
isomorphism $\7 C\otimes \2 H \cong \2 H \cong \2 H\otimes \7 C$.) 

Compact structures provide a formal way to bend wires --- indeed, we can now
connect a diagram represented with an operator with spectral decomposition $\sum_i
\beta_i\ket{i}\bra{i}$ bend all the open wires (or legs) towards the same direction
and it then can be thought of as representing a state ($\sum_i
\beta_i\ket{i}\ket{\overline{i}}$ where overbar is complex conjugation), bend them
the other way and it then can be thought of as
representing a measurement outcome ($\sum_i
\beta_i\bra{\overline{i}}\bra{i}$), that is an effect. (The
isomorphism $\sum_i \beta_i\bra{\overline{i}}\bra{i}\cong \sum_i
\beta_i\ket{i}\bra{i} \cong \sum_i
\beta_i\ket{i}\ket{\overline{i}}$ for a real valued basis becomes $\sum_i
\beta_i\bra{i}\bra{i}\cong \sum_i
\beta_i\ket{i}\bra{i} \cong \sum_i
\beta_i\ket{i}\ket{i}$ which amounts to flipping a bra to a ket and vise versa.)  One
can also connect inputs to outputs, contracting indices and creating larger and
larger networks. With these
ingredients in place let us now consider the new class of Boolean quantum states.

\subsection{Boolean and multi-valued tensor network states}\label{sec:c}
To illustrate the idea of defining Boolean and multi-valued logic gates as tensors,
consider Figure \ref{fig:andtensor} which depicts a simple but
key network building block: the use of the so-called ``quantum logic
\AND-tensor'' which we define in Section~\ref{sec:AND}.  This is a representation of
the familiar Boolean operation in the bit pattern of a tri-qubit quantum state as
\begin{equation*}
\ket{\psi_\AND} \bydef
\sum_{x_1,x_2\in\{0,1\}}\ket{x_1}\otimes\ket{x_2}\otimes\ket{x_1\wedge
x_2}=\ket{000}+\ket{010} +\ket{100} +\ket{111}
\end{equation*}
and hence the truth table of a function is encoded in the bit pattern of the
superposition state. This utilizes a linear representation
of Boolean gates on quantum states as opposed to the typical direct sum
representation common in Boolean algebra.

\begin{figure}[t]
\includegraphics[width=0.9\textwidth]{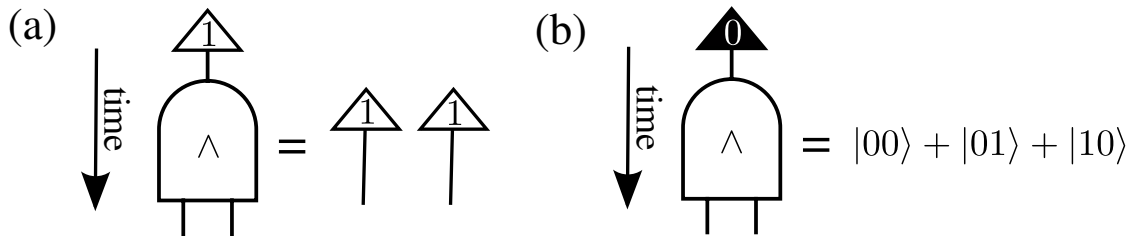}
\caption{Example of the Boolean quantum \AND-state or tensor.  In (a) the
network is run backwards (post-selected) to $\bra{1}$ resulting in
the product state $\ket{11}$.  In (b) the tensor is post-selected to
$\bra{0}$ resulting in the entangled state
$\ket{00}+\ket{01}+\ket{10}$.}\label{fig:andtensor}
\end{figure}

In this work we are particularly concerned with network constructions as a means 
to study many-body quantum states by tensor networks. First, we can compose
\AND-states (by connecting wires and hence contracting tensor indices) --- together
with \NOT-gates, this enables one to create the class of Boolean
states in Equation \eqref{eqn:Booleanstates}. That is, one will realize a network
that outputs
logical-one (represented here as $\ket{1}$) any time the input qubits represent a
desired term in a quantum state
(e.g.\ create a function that outputs logical-one on designated inputs
$\ket{00}$, $\ket{01}$ and $\ket{10}$ and zero otherwise as shown in Figure
\ref{fig:andtensor}). We then insert a $\ket{1}$
at the network output.  This procedure
recovers the desired Boolean state as illustrated in Figure \ref{fig:bs}(a) with the
resulting state appearing in Equation \eqref{eqn:Booleaninput}.  
\begin{equation} \label{eqn:Booleaninput}
\sum_{x_1,x_2,...,x_n\in\{0,1\}}\braket{1}{f(x_1
, x_2 ,...,x_n)}\ket{x_1,x_2,...,x_n}
\end{equation} 
The network representing the circuit is read backwards from output to
input.  Alternatively the full class of Boolean states is defined as:
\begin{definition}[Boolean many-body qudit states]
We define the class of Boolean states as those states which can be
expressed up to a global scalar factor in the form \eqref{eqn:Booleanstates}
\begin{equation}\label{eqn:Booleanstates}
\sum_{x_1,x_2,...,x_n\in\{0,1,...,d-1\}}\ket{x_1,x_2,...,x_n}\ket{f(x_1
, x_2 ,...,x_n)}
\end{equation} 
where $f:\7Z^n_d\rightarrow\7Z_d$ is a $d$-switching function and the sum is
taken to be over all variables $x_j$ taking 0 and 1 for qubits and $0, 1,...d-1$ in
the case of $d$-level qudits (see Figure \ref{fig:bs} (a)).  
\end{definition}

\begin{figure}[h]
\includegraphics[width=0.5\textwidth]{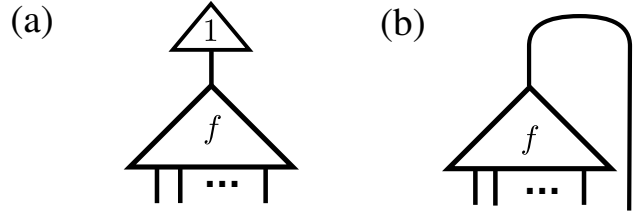}
\caption{A general multi-valued qudit state based on a $d$-switching function $f$ can
either be formed as (a) by inputting a logical-one at the output of the circuit as
described by Equation~\eqref{eqn:Booleaninput} or (b) by bending the output of the
circuit around to form an input as in Equation~\eqref{eqn:Booleanstates}.  For
multi-valued qudit states, the boolean function $f:\7 Z_2^n\rightarrow \7Z_2$ becomes
a multi-valued
qudit function $f:\7 Z^n_d\rightarrow \7Z_d$.  The network (a) is then
post-selected to $\alpha_0\ket{0}+\alpha_1\ket{1}+\cdots\alpha_{d-1}\ket{d-1}$ where
$\forall i, \alpha_i=0/1$. }\label{fig:bs}
\end{figure}

Examples of Boolean states include the familiar {\sf GHZ}-state
$\ket{00\cdots0}+\ket{11\cdots1}$ which on qudits in dimension $d$ becomes 
\begin{equation}
 \ket{{\sf GHZ}_d}= \sum_{i=0}^{d-1} \ket{i}\ket{i}\ket{i} =
\ket{0}\ket{0}\ket{0}+\ket{1}\ket{1}\ket{1}+\cdots+\ket{d-1}\ket{d-1}\ket{d-1} 
\end{equation}
as well as the \W-state
$\ket{00\cdots1}+\ket{01\cdots0}+\cdots+\ket{10\cdots0}$ which again on qudits
becomes (in Equation (\ref{eqn:w-stated}) the operator
$X\ket{m}=\ket{m+1(\text{mod}~d)}$ is
one way to define negation in higher dimensions.  The subscript labels the ket
(labeled 1,2 or 3 from left to right) the operator acts on $i$ times.) 
\begin{align}\label{eqn:w-stated}
\ket{\W_d} &:= \sum_{i=1}^{d-1} \sum_{j=1}^{3} (X_j)^i\ket{0}\ket{0}\ket{0} =
\ket{0}\ket{0}\ket{1}+\ket{0}\ket{1}\ket{0}+\ket{1}\ket{0}\ket{0}+\ket{0}\ket{0}\ket{
2}+\ket{0}\ket{2}\ket{0}+\ket{2}\ket{0}\ket{0}+\cdots\\\nonumber 
&\cdots+\ket{0}\ket{0}\ket{d-1}+\ket{0} \ket{d-1}\ket{0} +\ket{d-1}\ket{0}\ket{0} 
\end{align}
What
is clear from this definition is that Boolean states are always composed of equal
superpositions of sets of computational basis states, as the allowed scalars take binary values, 0,1. Our main result is that,
despite this apparent limitation,
tensor networks composed only of Boolean components can nonetheless describe
any quantum state. To do this we require a minor extension to include superposition
input/output states, e.g. rank-1 tensors of the form
$\ket{0}+\beta_1\ket{1}+\cdots+\beta_{d-1}\ket{d-1}$. This
gives a universal class of generalized Boolean tensor networks
which subsumes the important subclass of Boolean states.
This class is then shown to form a nascent example
of the exhaustiveness of CTNS and to give rise to a wide class of
quantum states that we show are exactly and efficiently sampled.  

\subsection{Putting it all together: connecting the dots}\label{sec:connectingdots}
The key point to this result is that the introduction of Boolean logic gate tensors
into the tensor network context
allows the seminal logic gate universality results from classical network
theory to be applied in the setting of tensor network states. By extending this result we can
construct a solution to the related
quantum problem --- that is, the decomposition or factorization of any quantum state
into a CTNS. Thus our main result is captured 
by the following statement (see Theorem~\ref{theorem:rep}).\\

\noindent \textbf{Result} (Translating quantum states into categorical tensor
networks). Given quantum state $\ket{\psi}$,
Theorem~\ref{theorem:rep} asserts a constructive method to factor $\ket{\psi}$ into a
CTNS constructed from rank-3, rank-2 tensors taken solely from the fixed set in the 
presentation from Figure~\ref{fig:F2-presentation} and \ref{fig:extraF2}
together with arbitrary rank-1 tensors.\\

This example then demonstrates the exhaustiveness of the most extreme case of the CTNS approach,
where almost all tensors are chosen from a small fixed set of tensors with precisely defined 
algebraic properties. Importantly, in Theorem \ref{corollary:states} the form of
this general construction is limited in such a way as
to provide a new class of states which can be exactly and efficiently sampled.

\section{constituent network components: a tensor tool box}\label{sec:components}

Any vector space $\2V$ has a dual $\2V^*$: this is the space of linear functions $f$
from $\2V$ to the ground field $\7C$, that is $f: \2V \rightarrow \7C$. This defines
the dual uniquely.  We must however fix a basis to identify the vector
space $\2V$ with its dual.  Given a basis, any basis vector $\ket{i}$ in $\2V$ gives
rise to a basis vector
$\bra{j}$ in $\2V^*$ defined by $\braket{j}{i} = \delta^j_i$ (Kronecker's delta). 
This
defines an isomorphism $\2V\rightarrow \2V^*$ sending $\ket{i}$ to $\bra{i}$ and
allowing us
to identify $\2V$ with $\2V^*$.  In what follows, we will fix a particular
arbitrarily chosen basis (called the computational basis in quantum information
science).  We will now concentrate on Boolean building blocks that are used
in our construction. 

\subsection{{\sf COPY}-tensors: the ``diagonal''}\label{sec:COPY}

The copy operation arises in digital circuits~\cite{Davio78, Weg87}
and more generally, in the context of category theory and Algebra, where it is
called a diagonal in cartesian categories. The
operation is readily defined in any finite dimension as
\begin{equation}
\bigtriangleup \bydef \sum_{i=0}^{d-1} \ket{ii}\bra{i}
\end{equation}
As $\ket{0}$ and
$\ket{1}$ are
eigenstates of $\sigma^z$, we might give $\bigtriangleup$ the alternative name of
{\sf Z}-copy.  In the
case of qubits {\sf COPY} is
succinctly presented by considering the map $\bigtriangleup$ that copies
$\sigma^z$-eigenstates:
$$
\bigtriangleup:\7C^2\rightarrow\7C^2\otimes\7C^2::
\begin{cases}
\ket{0} \mapsto \ket{00}\\
\ket{1} \mapsto \ket{11}
\end{cases}
$$
This map can be written in operator form as $\bigtriangleup:\ket{00}\bra{0} +
\ket{11} \bra{1}$ and
under cup/cap induced duality (on the right bra) this state becomes a {\sf GHZ}-state
as
$\ket{\psi_{\sf GHZ}} = \ket{000} + \ket{111}\cong\ket{00}\bra{0} +
\ket{11} \bra{1}$. The standard properties of {\sf COPY} are
given diagrammatically in Figure~\ref{fig:copygate} and a list of its relevant
mathematical properties are found in Table~\ref{fig:copygatesum}.

\begin{figure}[h]
\includegraphics[width=0.9\textwidth]{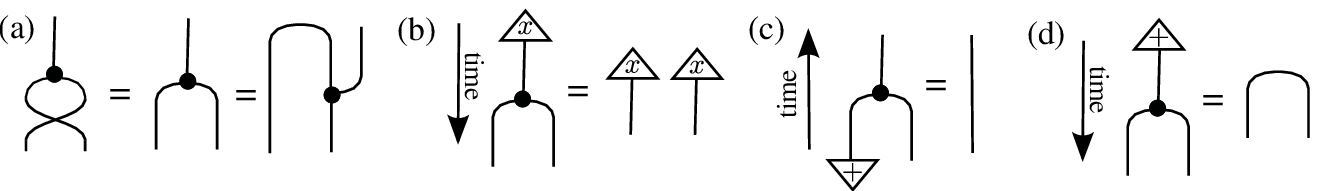}
\caption{Salient diagrammatic properties of the {\sf COPY}-dot. (a) Full-symmetry.
(b)
Copy points, e.g.\
$\ket{x}\mapsto\ket{xx}$ for $x=0,1$ for qubits and $x=0,1,...,d-1$ for $d$
dimensional qudits.  (c) The unit --- in this case
the unit corresponds to deletion, or a map to the
terminal object which is given as $\bra{+}\bydef\bra{0}+\bra{1}$ for qubits and
$\bra{+}\bydef\bra{0}+\bra{1}+\cdots+\bra{d-1}$ for $d$
dimensional qudits (the
bi-direction of time is
explained later by considering co-diagonals in Section~\ref{sec:co-diagonal}).  (d)
Co-interaction with the unit
creates a Bell state $\sum_{i=0}^{d-1}\ket{ii}$.  This and the corresponding dual
under the dagger form the compact structures of
the $\dagger$-category of
quantum theory.}\label{fig:copygate}
\end{figure}

\begin{remark}[The {\sf COPY}-gate from \CNOT]
The \CNOT-gate is defined as $\ket{0}\bra{0}_1\otimes \eye_2 +
\ket{1}\bra{1}_1\otimes \sigma^x_2$.  We will set the input that the target acts on
to $\ket{0}$ then calculate $\CNOT (\eye_1\otimes
\ket{0}_2)=\ket{0}\bra{0}_1\otimes \ket{0}_2 + \ket{1}\bra{1}_1\otimes
\ket{1}_2$.  We have hence defined the desired {\sf COPY} map copying states from the
Hilbert space with label $1$ (subscript) to the joint Hilbert space labeled $1$ and
$2$.
\end{remark}

\subsection{\XOR-tensors: the ``addition''}\label{sec:XOR}
The \XOR-gate implements exclusive disjunction or addition (mod 2 for
qubits)
and is denoted by the
symbol $\oplus$~\cite{Cohn62, xor70}. We note that for multi-valued
logic a modulo
subtraction gate can also be defined as in~\cite{BB11}.
By what could be called ``dot-duality'' the \XOR-gate is simply a Hadamard
transform of the {\sf COPY}-gate, appropriately applied to all of the dots
legs. (We denote the
discrete Fourier transform gate by~$H_{\hilbert{H}} := \frac{1}{\sqrt{d}}
\sum_{a,b\in\{0,1,...,d-1\}}e^{i 2 \pi a b/d}
\ketbra{a}{b}_{\hilbert{H}}$, where $d = \dim \hilbert{H}$ is the dimension of the
Hilbert space the gate
acts in. We can see that $H^T = H$, and that
in a qubit system $H$~coincides with the one-qubit Hadamard gate~\cite{BB11}.)  This
can be captured diagrammatically in the slightly different form: 
\begin{center}
\includegraphics[width=0.250\textwidth]{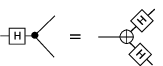}
\end{center}
To define the gate on the computational basis, we
consider $f(x_1,x_2)=x_1\oplus x_2$ then $f=0$ corresponds to
$(x_1,x_2)\in\{(0,0),(1,1)\}$ and $f=1$ corresponds to $(x_1,x_2)\in\{(1,0),(0,1)\}$,
where the
truth table for $\XOR$ follows
\begin{center}
\begin{tabular}{c|c|c}
$~x_1~$ & $~x_2~$ & $f(x_1,x_2)=x_1\oplus x_2$ \\ \hline
0 & 0 & 0 \\
0 & 1 & 1 \\
1 & 0 & 1 \\
1 & 1 & 0
\end{tabular}
\end{center}
Under cap/cap induced duality, the state defined by $\XOR$ is given as
\begin{equation}
\ket{\psi_\oplus}\bydef\sum_{x_1,x_2\in\{0,1\}}\ket{x_1}\ket{x_2}\ket{f(x_1,x_2)}
=\ket { 000 }
+\ket { 110 }
+\ket{011}+\ket{101}
\end{equation}
which is in the {\sf GHZ}-class by LOCC equivalence viz.
$\ket{\psi_\oplus}=\H\otimes\H\otimes\H(\ket{000}+\ket{111})$.  The operation of
\XOR~is
summarized in Table~\ref{fig:xorgatesum}.  Since the
\XOR-gate is related to the {\sf COPY}-gate by a change of basis, its diagrammatic
laws
have the same structure as those illustrated in Figure~\ref{fig:copygate}.  The gate
acting backwards (co-\XOR) is
defined on a basis as follows:
$$
\oplus:\7C^2\rightarrow\7C^2\otimes\7C^2::
\begin{cases}
\ket{0} \mapsto \ket{00}+\ket{11}\\
\ket{1} \mapsto \ket{10}+\ket{01}
\end{cases}~~~~\text{or equivalently}~~~~~~~\begin{cases}
\ket{+} \mapsto \ket{++}\\
\ket{-} \mapsto \ket{--}
\end{cases}
$$

\subsection{Generating the affine class of networks}
Thus far we have presented the {\sf XOR}- and {\sf COPY}- gates.  This system allows
us to create the linear class of Boolean functions.  As explained in the present
subsection, this class can be extended to to the affine class by introducing either a
gate that acts like an inverter, or by appending a constant $\ket{1}$ into our
system.  This constant will allow us to use the {\sf XOR}-gate to create an
inverter.   

A \textit{complemented Boolean variable} is a
Boolean variable that appears in negated form, that is $\neg x$ or
written equivalently as $\overline{x}$. Negation of a Boolean variable $x$ can be
expressed as the \XOR{} of
the variable with constant $1$ as $\overline{x}=1\oplus x$. Whereas
\textit{uncomplimneted Boolean
variables} are Boolean variables that do not appear in negated form (e.g.\ negation
is not allowed).  Linear Boolean functions contain terms with uncomplemented Boolean
variables that
appear individually (e.g.\ variable
products are not allowed such as $x_1x_2$ and higher orders etc., see
Section \ref{sec:Boolean}).  Linear Boolean functions take the general form
\be
f(x_1,x_2,...,x_n)=c_1x_1\oplus c_2x_2\oplus ...\oplus c_nx_n
\ee 
where the vector $(c_1,c_2,...,c_n)$ uniquely determines the function.  The affine
Boolean functions take the same general form as linear functions.  However, functions
in the affine class allows variables to appear in both complemented
and uncomplemented form. Affine Boolean functions take the general form
\begin{equation}\label{eqn:affine}
f(x_1,x_2,...,x_n)=c_0\oplus c_1x_1\oplus c_2x_2\oplus ...\oplus c_nx_n
\end{equation}
where $c_0=1$ gives functions outside the linear class.  From the identities,
$1\oplus 1=0$ and $0\oplus x=x$ we require the introduction of only one constant
($c_0$), see Appendix \ref{sec:Boolean}.

Together, \XOR{} and {\sf COPY} are not universal for classical circuits. When used
together, \XOR- and
{\sf COPY}-gates compose to create networks representing the class of linear
circuits. The affine circuits are generated by considering the constant $\ket{1}$.
The state $\ket{1}$ is
indeed copied by the black dot. However, our axiomatization
(Figure~\ref{fig:F2-presentation}) proceeds through 
considering the \XOR- and {\sf COPY}-gates together with $\ket{+}$, the unit for
{\sf COPY} and $\ket{0}$ the unit for \XOR.  It is by appending the constant
$\ket{1}$ into the
formal system (Figure~\ref{fig:F2-presentation}) that the affine class of circuits
can be realized.


\begin{remark}[Affine functions correspond to a basis]
Each affine function is labeled by a corresponding bit pattern.  This can be
thought of as labeling the computational basis, as
states of the form $\ket{\{0,1\}^n}$ are in correspondence with polynomials in
algebraic normal form (see Appendix \ref{sec:Boolean}).
\end{remark}

\subsection{Quantum \AND-state tensors: Boolean universality}\label{sec:AND}
The proceeding sections have introduced enough machinery to generate the linear and
affine classes of classical circuits.  These classes are not universal.  To recover
a universal system we must add a non-linear Boolean gate.  We do this by
representing the \AND{} gate as a tensor.  The unit for this gate is $\bra{1}$ and
so can be used to elevate the linear fragment to the affine class.  

The \AND{} gate (that is, $\wedge$) implements logical conjunction~\cite{Davio78,
Weg87}.  Using again ``dot-duality'', the \AND-gate relates to the \OR-gate via De
Morgan's law.
This can be captured diagrammatically as
\begin{center}
\includegraphics[width=0.35\textwidth]{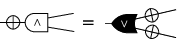}
\end{center}
To define the gate on the computational basis, we
consider $f(x_1,x_2)=x_1\wedge x_2$ which we write in short hand as $x_1 x_2$.  Here
$f=0$
corresponds to $(x_1,x_2)\in\{(0,0),(0,1),(1,0)\}$ and $f=1$ corresponds to
$(x_1,x_2)=(1,1)$.

Under cap/cap induced duality, the state defined by $\AND$ is given as
\begin{equation}\label{eqn:quantumAND}
\ket{\psi_\wedge}\bydef\sum_{x_1,x_2\in\{0,1\}}\ket{x_1}\ket{x_2}\ket{f(x_1,x_2)}
=\ket { 000 }
+\ket { 010 }
+\ket{010}+\ket{111}
\end{equation}
The key diagrammatic properties of {\sf AND}
are presented  in Figure~\ref{fig:ANDgate} and the gate is 
summarized in Table~\ref{fig:andgate}.  The gate
acting backwards (co-\AND) is
defined on a basis as follows: 
$$
\wedge:\7C^2\rightarrow\7C^2\otimes\7C^2::
\begin{cases}
\ket{0} \mapsto \ket{00}+\ket{01}+\ket{10}\\
\ket{1} \mapsto \ket{11}
\end{cases}~~~~\text{or}~~~~~~~\begin{cases}
\ket{+} \mapsto \ket{++}\\
\ket{-} \mapsto \ket{00}+\ket{01}+\ket{10} - \ket{11}
\end{cases}
$$

\begin{figure}[h]
\includegraphics[width=0.9\textwidth]{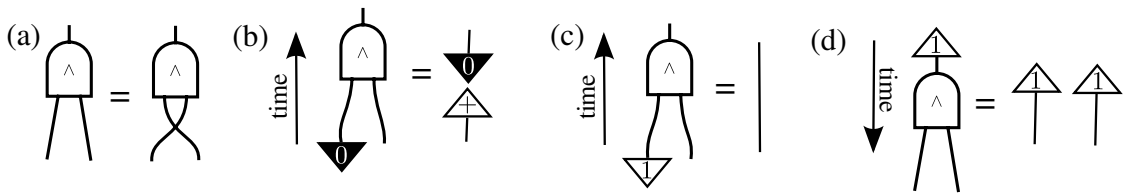}
\caption{Salient diagrammatic properties of the \AND-tensor. (a) Input-symmetry. (b)
Existence of a zero or fixed-point.
(c) The unit $\ket{1}$.  (d) Co-interaction with the unit
creates a product-state.  Note that the gate forms a valid
quantum operation when run backwards as in (d).}\label{fig:ANDgate}
\end{figure}

\begin{example}[\AND-states from Toffoli-gates]\label{ex:ANDfromtoff}
The \AND-state is readily constructed from the Toffoli gate as illustrated in
Figure~\ref{fig:ANDfromtoff}.  This allows some
interesting states to be created 
experimentally, for instance, post-selection of the output to $\ket{0}$
would yield the state $\ket{00}+\ket{01}+\ket{10}$. (See the course
notes~\cite{course} for more on how these
techniques can be used as an experimental prescription to generate quantum states. See also Figure \ref{fig:hadamard} for the connection to computationally universal quantum circuits.)
\end{example}

\begin{figure}[h]
\begin{center}
\includegraphics[width=0.80\textwidth]{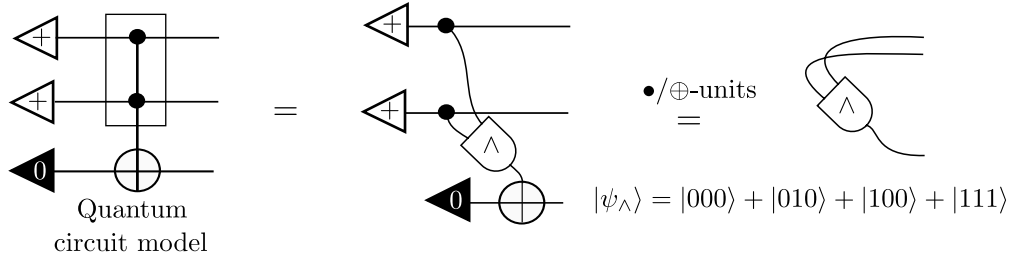}
\caption{Illustrates the use of units to prepare
the \AND-state.  Using
this state together with single qubit \NOT-gates, one can construct any Boolean
qubit state as well as any of the states appearing in Table~\ref{fig:BooleanStates}. 
We note that the box around the Toffoli gate (left) is meant to illustrate a
difference between our notation and that of quantum circuits.}\label{fig:ANDfromtoff}
\end{center}
\end{figure}


\begin{figure}[h]
\includegraphics[width=0.750\textwidth]{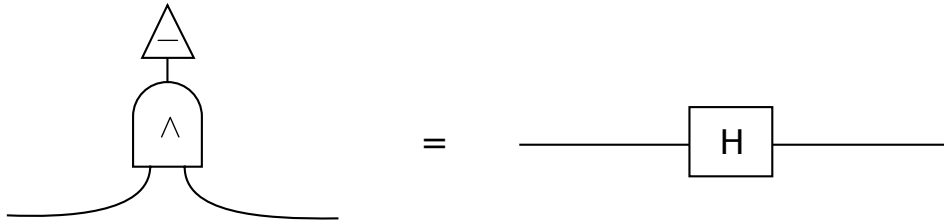}
\caption{Hadamard built from the \AND-state
together with $\ket{-}\bydef\frac{1}{\sqrt{2}}(\ket{0}-\ket{1})$.  We note
that quantum computational universality is already possible by considering simple Hadamard
states (e.g.\ $\ket{\psi_\H} = \ket{00}+\ket{01}+\ket{10}-\ket{11}$), {\sf COPY}-
and \AND-states, which follows from the proof that Hadamard and Toffoli are universal
for quantum circuits~\cite{A03}.}\label{fig:hadamard}
\end{figure}

\subsubsection{Summary of the {\sf XOR}-algebra on tensors}

We will now present the three previously referenced Tables (\ref{fig:copygatesum},
\ref{fig:xorgatesum} and \ref{fig:andgate}) which summarize the quantum logic tensors
we introduced in the previous subsections (\ref{sec:COPY}, \ref{sec:XOR} and
\ref{sec:AND}).  The tables contain entries listing properties that describe how
the introduced network components interact.  These interactions
are defined diagrammatically and explained in Section \ref{sec:interaction}.  

\begin{center}
\begin{table}[h]
\begin{tabular}{|c|c|c|c|}\hline
Gate Type & Co-copy point(s) & Unit & Co-unit Interaction\\\hline
{\sf COPY} & $\ket{0}$,$\ket{1}$  & $\ket{+}$  & Bell state: $\ket{00}+\ket{11}$
 \\ \hline\hline
Symmetry & Associative & Commutative & Frobenius Algebra\\\hline
Full  & Yes         & Yes & Yes (Node Equivalence) \\ \hline
\end{tabular}
\caption{Summary of the {\sf COPY}-gate from Section
\ref{sec:COPY}.}\label{fig:copygatesum}
\end{table}
\end{center}

\begin{center}
\begin{table}[h]
\begin{tabular}{|c|c|c|c|}\hline
Gate Type & Co-copy point(s) & Unit & Co-unit Interaction\\\hline
\XOR & $\ket{+}$,$\ket{-}$  & $\ket{0}$  & Bell state: $\ket{00}+\ket{11}$
 \\ \hline\hline
Symmetry & Associative & Commutative & Frobenius Algebra\\\hline
Full  & Yes         & Yes & Yes (Node Equivalence) \\ \hline
\end{tabular}
\caption{Summary of the \XOR-gate from Section \ref{sec:XOR}.}\label{fig:xorgatesum}
\end{table}
\end{center}

\begin{center}
\begin{table}[h]
\begin{tabular}{|c|c|c|c|}\hline
Gate Type & Co-copy point(s) & Unit & Co-unit Interaction\\\hline
\AND & $\ket{1}$  & $\ket{1}$  & Product state: $\ket{11}$
 \\ \hline\hline
Symmetry   & Associative & Commutative & Bialgebra Law\\\hline
Inputs  & Yes         & Yes         & Yes (with {\sf GHZ}) \\ \hline
\end{tabular}
\caption{Summary of the \AND-gate from Section \ref{sec:AND}.}\label{fig:andgate}
\end{table}
\end{center}

\subsection{co-{\sf COPY}: the co-diagonal}\label{sec:co-diagonal}
What is evident from our subsequent discussions on logic gates is that in the context
of tensors, the bending of wires implies that gates can be used both forwards in backwards. We can therefore form tensor networks from Boolean gates
in a very
different way from classical circuits. Indeed, it becomes possible to flip a {\sf
COPY} operation upside down,
that is, instead of having a single leg split into two legs, have two legs merge into
one.  In
terms of tensor networks, co-{\sf COPY} is simply
thought of as being a dual (transpose) to the familiar {\sf COPY} operation. This is
common in
algebra: to consider the dual
notation to algebra, that is co-algebra.  In general, while a product is a joining or pairing (e.g.\ taking two
vectors and producing a third) a co-product is a co-pairing taking a single vector
in the space $\2A$ and producing a vector in the space $\2A\otimes \2 A$.

\begin{remark} [co-algebras~\cite{FA}]
co-algebras are structures that are dual (in the sense of
reversing arrows) to unital associative algebras such as {\sf COPY} and
\AND{} the axioms of which we formulated in terms of picture calculi
(Sections \ref{sec:COPY} and \ref{sec:AND}).  Every co-algebra, by
(vector space) duality, gives rise to an algebra, and in finite dimensions, this
duality goes in both directions.
\end{remark}

Co-{\sf COPY} can be thought of as applying a delta function in the transition from
input to output.  That is, given a copy point $x=0,1,...,d-1$ for qudits on dim $d$.
 Depicting {\sf COPY} as the map
$\bigtriangleup$
\begin{equation}
\bigtriangleup(\ket{x})=\ket{x}\otimes\ket{x}
\end{equation}
we define co-{\sf COPY} by the map $\bigtriangledown$ such that 
\begin{equation}
\bigtriangledown(\ket{i},\ket{j})=\delta_{ij} \ket{i}
\end{equation}
that is, the diagram is mapped to zero (or empty) if the inputs $\ket{i}$, $\ket{j}$
do not agree.  This is succinctly expressed in terms of a delta-function dependent on
inputs $\ket{i}$, $\ket{j}$ where $i,j=0,1,...,d-1$ for qudits of dim $d$.

\begin{example}[Simple co-pairing]
Measurement effects on tripartite quantum systems can be thought of as
co-products.  This is given as a map from one system (measuring the first) into two
systems (the effect this has on the other two). {\sf GHZ}-states are prototypical
examples of co-pairings.  In this case, the measurement outcome of $\ket{0}$
($\ket{1}$) on a single subsystem sends the other qubits to $\ket{00}$ ($\ket{11}$)
and by linearity this sends $\ket{+}$ to $\ket{00}+\ket{11}$.  
\end{example}

\subsection{The remaining Boolean tensors: \NAND-states etc.}
We have represented a logical system on tensors --- this enables us
to represent any Boolean function as a connected network of tensors and hence any Boolean state.
We chose as our generators,
constant $\ket{1}$, {\sf COPY}, \XOR, \AND{}.  Other generators could have also been
chosen such as \NAND-tensors. Our choice however, was made as a matter of convenience. If we
had considered other generators, we could have ended up considering the
following cases: weak-units~(Definition \ref{def:weakunits}) and fixed point
pairs~(Definition \ref{def:fixedpoints}). We note that the
\NAND-states were used in \cite{1996quant.ph..5011S} for fault-tolerant quantum
computation --- see also \cite{toffoli-state}.

\begin{definition}[Weak units]\label{def:weakunits}
An algebra (or product see Appendix \ref{sec:newalgebra}) on a tripartite state
$\ket{\psi}$ has a unit (equivalently, one has that the state is unital) if there
exists an effect $\bra{\phi}$ which
the product acts on to produce an invertible map $B$, where $B=\eye$ (see
Example~\ref{ex:weakunits}). If no such $\bra{\phi}$ exists to make $B=\eye$, and $B$
has an inverse, we call $\bra{\phi}$ a weak unit, and say the state $\ket{\psi}$ is
weak unital and if $B\neq \eye$ and $B^2=1$ we call
the algebra on $\ket{\psi}$ unital-involutive.  This scenario is given
diagrammatically
as:
\begin{center}
\includegraphics[width=0.50\textwidth]{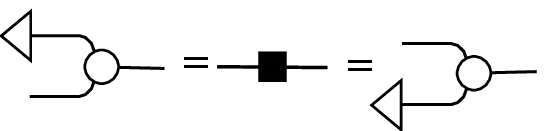}
\end{center}
\end{definition}
\begin{example}[\NAND{} and \NOR{}]\label{ex:weakunits}
\NAND{} and \NOR{} have weak units, respectively given by $\ket{1}$ and $\ket{0}$.
These weak units are unital-involutive.
\begin{equation}
 \ket{\psi_\NAND} = \ket{001}+\ket{011}+\ket{101}+\ket{110}
\end{equation}
\begin{equation}
 \ket{\psi_\NOR} = \ket{001}+\ket{010}+\ket{100}+\ket{110}
\end{equation}
For $\ket{\psi_\NAND}$ to have a unit, there must exist a $\ket{\phi}$ such that
\begin{equation}
 \braket{\overline{\phi}}{0}\ket{01}+\braket{\overline{\phi}}{0}\ket{11}+\braket{
\overline{\phi}}{1}\ket{01} +\braket{\overline{\phi}}{1}\ket{10}=\ket{00}+\ket{11}
\end{equation}
and hence no choice of $\ket{\phi}$ makes this possible, thereby confirming the
claim.
\end{example}

\begin{definition}[Fixed Point Pair]\label{def:fixedpoints}
An algebra (see Appendix \ref{sec:newalgebra}) on a tripartite state
$\ket{\psi}$ has a \textit{fixed point} if
there exists an effect $\bra{\phi}$ (the fixed point) which
the product acts on to produce a constant output, independent of the other input value. 
For instance, in Figure \ref{fig:fixedpoints}(c) on the left hand side the effect
$\bra{1}$ induces a map (read bottom to top) that sends $\ket{+}\mapsto \ket{1}$. 
Up to a scalar, this map expands linearly sending both basis effects $\bra{0}$,
$\bra{1}$ to to the constant state $\ket{1}$. If the resulting output is the same as
the fixed point, we say $\bra{\phi}$ has a zero ($\ket{1}$ is the zero for the
\OR-gate in Figure \ref{fig:fixedpoints}(c)). A
fixed point pair consists of two algebras with fixed points, such that the fixed
point of one algebra is the unit of the other, and vise versa (see
Figure~\ref{fig:fixedpoints2} and \ref{fig:fixedpoints}). Diagrammatically this
is given in Figure \ref{fig:fixedpoints2}.  
\end{definition}

\begin{figure}[h]
\includegraphics[width=0.750\textwidth]{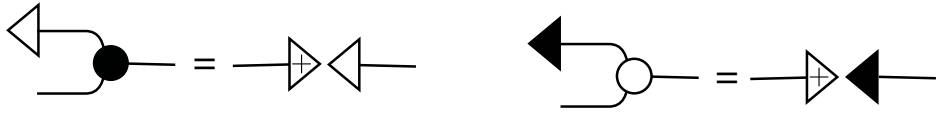}
\caption{Diagrammatic equations satisfied by a fixed point
pair (see Definition \ref{def:fixedpoints}). }\label{fig:fixedpoints2}
\end{figure}

\begin{figure}[h]
\includegraphics[width=0.8\textwidth]{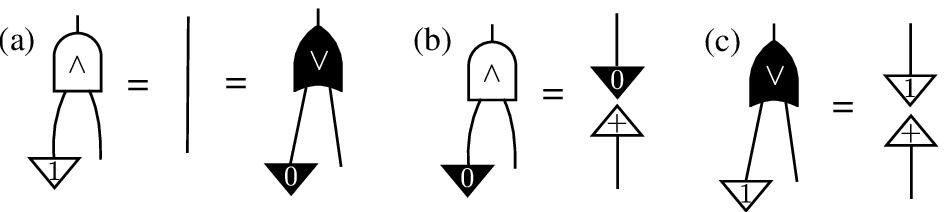}
\caption{\AND{} and \OR{} tensors form a fixed point pair.  The unit for \AND{}
($\ket{1}$
see a) is the zero for \OR{} (c) and vise versa: the unit of \OR{} ($\ket{0}$
see a) is the zero for \AND{} (b).}\label{fig:fixedpoints}
\end{figure}

\subsection{Summarizing: network composition of quantum logic tensors}
We have considered sets of universal classical structures in our tensor
network model.  In classical computer science, a universal set of gates is able to
express any $n$-bit Boolean function
\begin{equation}
  f:\7B^n\rightarrow\7B::(x_1,...,x_n)\mapsto f(x_1,...,x_n)
\end{equation}
where we note that $\7Z_2\cong\7B$ allowing us to use the alternative notation for
$f$ as $f:\7Z_d^n\rightarrow\7Z_d$ with $d=2$ for the binary case. Universal sets
include \{{\sf COPY}, \NAND\}, \{{\sf COPY}, \AND,~\NOT\}, \{{\sf COPY},
\AND,~\XOR,$\ket{1}$\}, \{\OR,~\XNOR,$\ket{1}$\}
and others.  One can also consider the states $\ket{\psi}$ formed by the bit patterns
of
these functions $f(x_1,x_2)$ as
\be
\ket{\psi_f} = \sum_{x_1,x_2\in\{0,1\}}\ket{x_1}\ket{x_2}\ket{f(x_1,x_2)}
\ee
This allows a wide class of states to be constructed effectively.  In the
following
Table (\ref{fig:BooleanStates}) we illustrate the quantum states representing the classical
function of two-inputs.

\begin{center}
\begin{table}[h]
\begin{tabular}{|c|c|}\hline
non-linear & linear (Frobenius Algebras) \\\hline
$\ket{\psi_\AND}=\ket{000}+\ket{010}+\ket{100}+\ket{111}$ &  \\
~$\ket{\psi_\OR}=\ket{001}+\ket{011}+\ket{101}+\ket{111}$ &
~$\ket{\psi_\XOR}=\ket{000}+\ket{011}+\ket{101}+\ket{110}$\\
$\ket{\psi_\NAND}=\ket{001}+\ket{011}+\ket{101}+\ket{110}$ &
$\ket{\psi_\XNOR}=\ket{001}+\ket{010}+\ket{100}+\ket{111}$\\
~$\ket{\psi_\NOR}=\ket{001}+\ket{010}+\ket{100}+\ket{110}$ & \\\hline
\end{tabular}
\caption{The bit pattern of these quantum states represents a Boolean function (given
by the subscript) such that the right most bit is the Boolean functions output, and
the two left bits are the functions inputs, and the non-linear Boolean functions are
on the left side of the table and the linear functions on the right.  Consider the
state $\ket{\psi_\AND}$, and Boolean variables $x_1$ and $x_2$, then the
superposition
$\ket{\psi_\AND}$ encodes the function $\ket{x_1,x_2,x_1\wedge x_2}$ in each term in
the
superposition, and $\ket{\psi_\AND}=\sum_{x_1,x_2\in\{0,1\}}\ket{x_1,x_2,x_1\wedge
x_2}$.
As outlined in the text, cup/cap induced-duality allows us (for instance) to express
this state as the operator
$\ket{0}\bra{00}+\ket{0}\bra{01}+\ket{0}\bra{01}+\ket{1}\bra{11}::\ket{x_1,x_2}
\mapsto\ket{x_1\wedge x_2}$ which projects qubit states to the $\AND$ of their bit
value.}\label{fig:BooleanStates}
\end{table}
\end{center}

\section{Interaction of the network components}\label{sec:interaction}
Having outlined the Boolean components used in our tensor toolbox we now
explore how these tensors interact when connected in a tensor network. The
interactions can be defined diagrammatically and given simple rewrite rules for CTNS
based on these component tensors.

\subsubsection{Merging {\sf COPY}-dots by node equivalence}
{\sf COPY}-dots are readily generalized to an arbitrary number of input and
output legs. As one would rightly suspect,
a {\sf COPY}-dot with $n$~inputs and $m$~outputs corresponds to an
$n+m$-partite {\sf GHZ}-state.  Neighboring dots of the same type can be merged into
a single dot: this is called node equivalence in digital circuits. {\sf
COPY}-dots represent Frobenius algebras~\cite{Carboni-Walters,FA, redgreen}.  (Note that the recent online version 
of \cite{redgreen} was already influenced by the arXiv version of the present manuscript as well as \cite{BB11, course}.) 

\begin{theorem}[Node equivalence or spider
law]\label{theorem:spider} Given a
connected graph
with $m$~inputs and $n$~outputs comprised solely of \COPY{} dots of
equal dimension, this map
can be equivalently expressed as a single $m$-to-$n$ dot, as shown
in Figure~\ref{fig:Spider}.
\end{theorem}

\begin{example}[Two-site reduced density operator of $n$-party {\sf GHZ}-states]
{\sf GHZ}-states on $n$-parties have a well known matrix product expression given as
\be
\ket{{\sf GHZ}_n} = \text{Tr}\left[ \left( \begin{array}{ccccc}
\ket{0} & 0 & \cdots & 0 \\
0 & \ket{1} & \cdots & 0 \\
\vdots & \cdots        & \ddots & \vdots \\
0 &  \cdots & 0 & \ket{d-1} 
\end{array} \right)^{n}\right]=
\ket{0}\ket{0}\ket{0}+\ket{1}\ket{1}\ket{1}+\cdots+\ket{d-1}\ket{d-1}\ket{d-1}
\ee
Such MPS networks are
known to be efficiently contactable.  We note that the networks in
Figure~\ref{fig:Spider} do not appear \textit{a priori} to be contractible due to the
number of open legs. What makes them contractible (in their present from) is that
the tensors obey node equivalence
law allowing them to be deformed into a contractible MPS network: see Figure~\ref{fig:Spider-comb}.  The reduced density matrix of an n-party
{\sf GHZ}-state then becomes (a) in
Figure~\ref{fig:ReducedGHZ} and the expectation value of an observable is shown in
(b) where we included the normalisation constant.
\end{example}

\begin{figure}[h]
\includegraphics[width=0.6\textwidth]{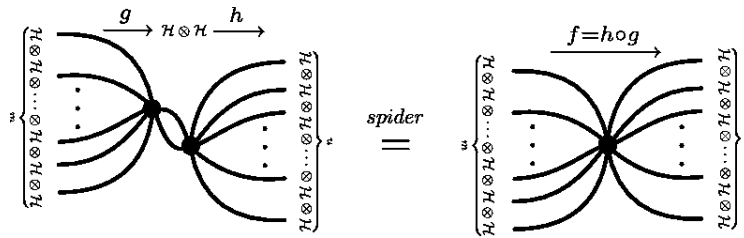}
\caption{Node equivalence or spider law \cite{FA}.  Connected black-dots ($\bullet$) as well
as
connected plus-dots
($\oplus$) can be merged and also split apart at will.  The intuition for digital or
qudit circuits follows by connecting a state $\ket{\phi}$ to one of the legs and
iterating over a complete basis $\ket{0}$, $\ket{1}$,...,$\ket{d-1}$.}
\label{fig:Spider}
\end{figure}

\begin{figure}[h]
\includegraphics[width=0.7\textwidth]{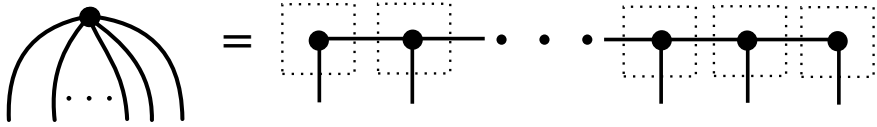}
\caption{The {\sf GHZ}-state tensor is simply a rank-$n$ {\sf COPY}-dot.  Node
equivalence
implies that this tensor can be deformed into any network
geometry including a MPS comb-like structure (right).}
\label{fig:Spider-comb}
\end{figure}

\begin{figure}[h]
\includegraphics[width=0.9\textwidth]{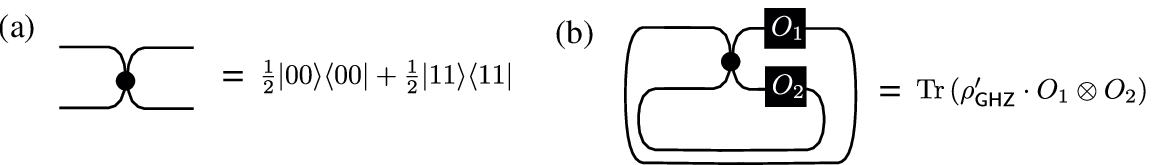}
\caption{Reduced density operator.  Left (a) reduced density operator $\rho_{\sf
GHZ}'$
found from applying the node equivalence law to a n-qubit {\sf GHZ}-state.  Right (b) the
expectation value of observable $O_1\otimes O_2$ found from connecting the
observable and connecting the open legs (i.e.\ taking the trace).}
\label{fig:ReducedGHZ}
\end{figure}

\subsection{Associativity, distributivity and commutativity}

The products we have considered are all associative and commutative.  As algebras,
\AND{}, \XOR{} and {\sf COPY}{} are associative, unital commutative algebras.  This
was
already expressed diagrammatically in Figures~\ref{fig:F2-presentation}(a) and
Figure \ref{fig:extraF2}(c).
The diagrammatic laws relevant for this subsection represent the following
Equations
\begin{equation}
(x_1\wedge x_2)\wedge x_3= x_1\wedge (x_2\wedge x_3)
\end{equation}
\begin{equation}
(x_1\oplus x_2)\oplus x_3= x_1\oplus (x_2\oplus x_3)
\end{equation}
Distributivity of \AND{} over \XOR{} then becomes (see (h) in
Figure~\ref{fig:extraF2})
\begin{equation}
(x_1\oplus x_2)\wedge x_3= (x_1\wedge x_2)\oplus (x_1\wedge x_2)
\end{equation}
We have commutativity for any product symmetric in its inputs: this is
the case for \AND{} and \XOR.

\subsection{Bialgebras on tensors}\label{sec:bialgebra}
There is a powerful type of algebra that arises in our
setting: a bialgebra defined graphically on tensors in Figure
\ref{fig:BialgebraAxioms} (see Kassel, Chapter III~\cite{Kassel}, \cite{FA}).  

 Such an algebra is simultaneously a unital associative algebra and co-algebra (for
the associativity condition see (b) in Figure~\ref{fig:BialgebraAxioms}).
Specifically, we consider the following two ingredients:
\begin{description}\addtolength{\itemsep}{-0.5\baselineskip}
\item[(i)] a product (black dot) with a unit (black triangle) see the right hand
side of Figure~\ref{fig:BialgebraAxioms}(a) 
\item[(ii)] a co-product (white dot) with a co-unit
(white triangle) see the left hand side of Figure~\ref{fig:BialgebraAxioms}(a) 
\end{description}
To form a bialgebra, these two ingredients above must be characterized by
the following four compatibility conditions:
\begin{description}\addtolength{\itemsep}{-0.5\baselineskip}
\item[(i)] The unit of the black dot is a copy-point of the white dot as in (e)
from Figure~\ref{fig:BialgebraAxioms}.
\item[(ii)] The (co)unit of the white dot is a copy-point of the black dot as in (d)
from Figure~\ref{fig:BialgebraAxioms}.
\item[(iii)] The bialgebra-law is satisfied given in (c) from
Figure~\ref{fig:BialgebraAxioms}.
\item[(iv)] The inner product of the unit (black triangle) and the co-unit (white
triangle) is non-zero (not shown in Figure~\ref{fig:BialgebraAxioms}).
\end{description}

\begin{figure}[h]
\includegraphics[width=1.0\textwidth]{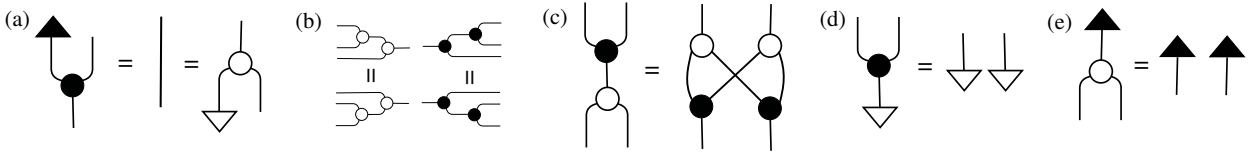}
\caption{Bialgebra axioms~\cite{FA} (scalars are omitted).  (a) unit laws (these are
of course
left and right units); (b) associativity; (c) bialgebra; (d,e) co-{\sf COPY}
points.}\label{fig:BialgebraAxioms}
\end{figure}

\begin{example}[{\sf GHZ}, \AND{} form a bialgebra]
We are in a position to study the interaction of {\sf GHZ}-\AND.   This interaction
satisfies the equations in Figure \ref{fig:BialgebraAxioms}: (a) the bialgebra law;
(b) the co-copy point of \AND{} is $\ket{1}$; and (c) the co-interaction with the
unit for {\sf GHZ} creates a compact structure.  In addition, (a) and (b) show the
copy points for the black {\sf GHZ}-dot; in (c) we have the unit and fixed point
laws.
\end{example}

Even if a given product and co-product do not satisfy all of the compatibility
conditions (given in (a), (b), (c), (d), (e) in Figure \ref{fig:BialgebraAxioms}),
and hence do not form bialgebras, they can still satisfy the bialgebra law which is
given in Figure \ref{fig:BialgebraAxioms}(c).  Examples of states
that satisfy the bialgebra law in Figure \ref{fig:BialgebraAxioms}(c), but are not
bialgebras are given in Definition~\ref{def:bialgebra}. Notice that bialgebra
provides a highly constraining characterization of the tensors involved and is
tantamount to defining a commutation relation between them.

\begin{definition}[Bialgebra Law~\cite{FA}]\label{def:bialgebra}
A pair of quantum states (black, white dots) satisfy the bialgebra law if (c)
in Figure \ref{fig:BialgebraAxioms} holds.  The Boolean states, \AND, \OR, \XOR,
\XNOR, \NAND, \NOR{} all satisfy the bialgebra law with {\sf COPY}.  
\end{definition}

\subsubsection{Hopf algebras on tensors}\label{sec:hopf}
A particularly important class of bialgebras are known as Hopf-algebras~\cite{FA}.
This is characterized by the way in which algebras and co-algebras can interact.  This
is captured by the Hopf-law, where the linear map $A$ is known as the antipode.

\begin{definition}[Hopf-Law~\cite{FA}]
A pair of quantum states satisfy the Hopf-Law if an $A$ can be found such that the
following equations hold:
\begin{center}
\includegraphics[width=.80\textwidth]{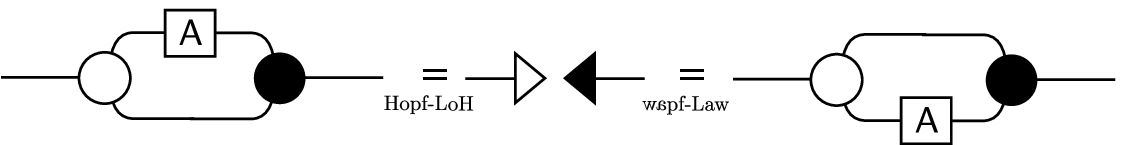}
\end{center}
\end{definition}

\begin{example}[\XOR{} and {\sf COPY} are Hopf-algebras on Boolean States
\cite{Boolean03}]
It is well known (see e.g.\ \cite{Boolean03}) that the Boolean state \XOR, satisfies
the Hopf-algebra law with trivial antipode ($A=\eye$) with {\sf COPY}.  Recall Figure
\ref{fig:F2-presentation}(g). 
\end{example}

\subsection{Bending wires: compact structures}\label{sec:bends}
As mentioned in the preliminary section (\ref{sec:overview}), we make use of
what's called a \textit{compact structure} in category theory which amounts to
introducing cups and caps, to provide a formal way to bend wires and define
transposition. See Figures
\ref{fig:cupsetc} and \ref{fig:adjoints}.  

A \textit{compact structure} on an object $\2H$ consists of another object $\2H^*$
together with a pair of morphisms (note that we use the equation $\2H^*=\2H$ in
Hilbert space making objects self dual which simplifies what follows).  The theory is 
well known in the modern mathematics of category theory and has been used to study teleportation in~\cite{catprotocols}. 
Consider 
\begin{equation*}
  \eta_{\2H} : \eye \longrightarrow {\2H} \otimes {\2H} ~~~~~~~~~~~~~~\epsilon_{\2H}
: {\2H} \otimes {\2H} \longrightarrow \eye
\end{equation*}
where the standard representation in Hilbert space with dimension $d$ and basis
$\{\ket{i}\}$ is given by
\begin{equation*}
\eta_{\2H} = \sum_{i=0}^{d-1}\ket{i}\otimes\ket{i} ~~~~~~~~~~~~~~ \epsilon_{\2H} =
\sum_{i=0}^{d-1}\bra{i}\otimes\bra{i}
\end{equation*}
and in string diagrams (read from the top to the bottom of the page) as
\begin{center}
  \includegraphics[width=0.35\textwidth]{D4.eps}
\end{center}
These cups and caps give rise to cup/cap-induced
duality: this amounts to being able to create a linear map that ``flips'' a bra to a
ket (and vise versa) and at the same time taking an (anti-linear) complex conjugate.
 In other words, the cap $\sum_{i=0}^1\bra{ii}$ sends  quantum state
$\ket{\psi}=\alpha\ket{0} + \beta\ket{1}$ to $\alpha\bra{0} + \beta\bra{1}$ which is
equal to the complex conjugate of $\ket{\psi}^\dagger =
\bra{\psi}=\overline{\alpha}\bra{0} + \overline{\beta}\bra{1}$.  Diagrammatically,
the dagger is given by mirroring operators across the page, whereas transposition is
given by bending wire(s). Clearly, $\bra{\overline{\psi}}=\alpha\bra{0} +
\beta\bra{1}$.  

In the case of relating the Bell-states and effects to the identity operator, under
cup/cap-induced duality, we flip the second ket on $\eta_{\2H}$ and the first bra on
$\epsilon_{\2H}$.  This relates these maps and the identity $\eye_{\2H}$ of the
Hilbert space: that is, we can fix a basis and construct
invertible maps sending $\eta_{\2H}~\cong~\eye_{\2H}~\cong~\epsilon_{\2H}$.  More
generally, the maps $\eta_{\2H}$ and $\epsilon_{\2H}$ satisfy the
equations given in Figure \ref{fig:cupsetc} and their duals
under the dagger. 

A second way to introduce cups and caps is to consider a \textit{Frobenius
form}~\cite{FA} on
either of the structures in the linear fragment from Figure \ref{fig:F2-presentation}
({\sf COPY} and \XOR).  This is simply a functional that turns a
product/co-product into a cup/cap.  This allows one to recover the above compact
structures (that is, the cups and caps given above) as
\begin{center}
  \includegraphics[width=0.5\textwidth]{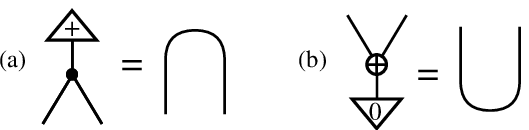}
\end{center}
Again, we will use these cups and caps as a formal way to bend wires in tensor
networks: this can be thought of simply as a reshape of a matrix.

\begin{figure}[h]
  \includegraphics[width=0.7\textwidth]{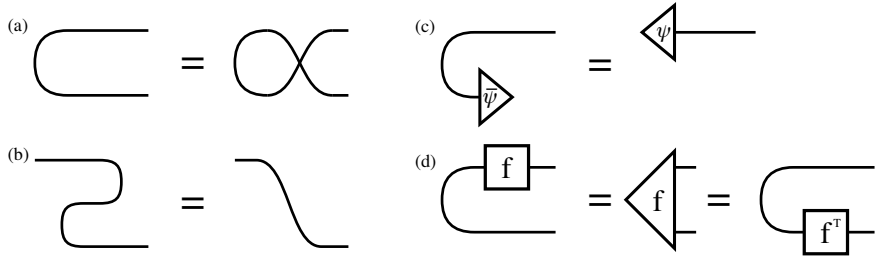}
\caption{Cup identities. (a) Symmetry. (b) Conjugate state.
(c) Teleportation \cite{catQM} or the \textit{snake equation}.
(d) Sliding an operator around a cup transposes it.}\label{fig:cupsetc}
\end{figure}

\begin{figure}[h]
  \includegraphics[width=0.7\textwidth]{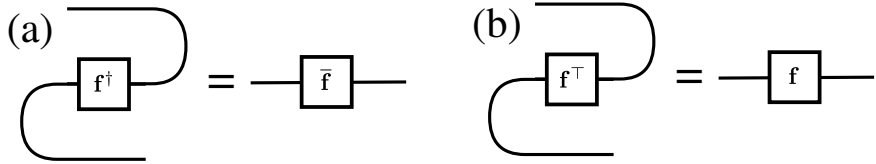}
\caption{Diagrammatic adjoints.  Cups and caps allow
us to take the transpose of a
linear map.  Note that care must be taken, as flipping
a ket $\ket{\psi}$ to a bra $\bra{\psi}$ is conjugate
transpose, and bending a wire is simply transposition, so the conjugate must be
taken: e.g.\ acting on $\ket{\psi}$ with a cap given as $\sum_i \bra{ii}$ results in
$\bra{\overline \psi}$.  }\label{fig:adjoints}
\end{figure}


\section{Examples of Categorical tensor network states}\label{sec:CTN}
Our categorical
approach enables one to translate
a quantum state directly into a new type of network: a so-called CTNS. We have
focused on Boolean
network components and have already presented in detail their algebraic properties and defining characteristics.
Here we will illustrate their expressive power by considering a few elementary examples before presenting our
main theorem \eqref{theorem:rep}, precisely showing how to determine a categorical
tensor network to represent any given quantum state.

\subsection{Constructing Boolean states}\label{sec:W}
Since the fixed building blocks of our tensor networks are the logic tensors \AND,
\OR,
\XOR\, and {\sf COPY}, along
with ancilla bits, we can immediately apply the universality of these elements for classical circuit construction
to guarantee that any Boolean state has a categorical tensor network decomposition.
However our construction goes beyond this because as we have seen,
categorical tensor networks can be deformed and rewired in ways which are not
ordinarily permitted in the standard acyclic-temporal definition of classical
circuits. The \W-state will be shown to provide a non-trivial example of this.

\begin{example}[Functions on \W- and {\sf GHZ}-states]\label{ex:WandGHZ}
We consider the function $f_\W$ which outputs logical-one given input bit string
$001$, $010$ and $100$ and logical-zero otherwise.  Likewise the function
$f_{{\sf GHZ}}$ is defined to output logical-one on input bit strings $000$ and $111$
and
logical-zero otherwise.  See Examples \ref{ex:WandGHZBoolean} and
\ref{ex:mlWandGHZ} which
consider representation of these functions as polynomials.  We will continue to work
with a linear representation of functions on quantum states; here
bit string
$000\mapsto \ket{000}$ (etc.).
\end{example}

\begin{example}[MPS form for \W-state]
Like the {\sf GHZ} state, the \W-state has a
simple MPS representation
\be\label{eqn:wnMPS}
\ket{\W_n} = \bra{0}\left( \begin{array}{cc}
\ket{0} & 0 \\
\ket{1} & \ket{0} \end{array} \right)^{n}\ket{1}=
\ket{10...0}+\ket{01...0}+...+\ket{00...1}.
\ee
This description \eqref{eqn:wnMPS} is succinct.   All MPS-states have essentially this 
same topological or network structure. In contrast, our
categorical construction
described below breaks this network up further.
\end{example}

\begin{remark}[Exact-value functions]
The function $f_\W$ takes value logical-one on input vectors with $k$ ones for a
fixed integer $k$.
Such functions are known in the literature as Exact-value symmetric Boolean
functions.  When cast into our framework, exact-value functions give
rise to tensor networks which represent what are known as Dicke states
\cite{2010NJPh...12g3025A}. 
\end{remark}

\begin{example}[Function realization of $f_\W$ and
$f_{{\sf GHZ}}$: the Boolean case]\label{ex:WandGHZBoolean}
One can express (using $\overline x$ to mean Boolean variable negation)
\begin{equation}
f_\W(x_1,x_2,x_3)=\overline x_1\overline x_2x_3\oplus
x_1\overline x_2\overline x_3\oplus \overline x_1x_2\overline x_3
\end{equation}
by noting that each term in the disjunctive normal form of $f_\W$ are disjoint, and
hence \OR~maps to \XOR~as $\vee\mapsto \oplus$.  The algebraic normal form (see
Appendix~\ref{sec:Boolean}) becomes
\begin{equation}\label{eqn:fWBoolean}
f_\W(x_1,x_2,x_3)=x_1 \oplus x_2\oplus x_3\oplus x_1 x_2 x_3
\end{equation}
\begin{equation}
f_{{\sf GHZ}}(x_1,x_2,x_3) =1\oplus x_1\oplus x_2\oplus x_3\oplus x_1 x_2\oplus
x_1x_3\oplus x_2x_3
\end{equation}
\end{example}

\begin{example}[Function realization of $f_\W$ and $f_{{\sf GHZ}}$: the set
function case] \label{ex:mlWandGHZ}
Set functions are mappings from the family of subsets of a finite ground set (e.g.\
Booleans) to the real or complex numbers.  In the circuit theory literature,
functions from the Booleans to the reals are known as pseudo-Boolean functions and
more commonly as
multi-linear polynomials or forms (see~\cite{JDB08} where these functions are used to
embed a co-algebraic theory of logic gates in the ground state energy configuration of
spin models). There exists an algebraic normal form and hence a unique multi-linear
polynomial representation for each pseudo-Boolean
function (see Appendix~\ref{sec:Boolean}).  This is found by mapping the negated
Boolean variable as
$\overline{x}\mapsto
(1-x)$.  For the {\sf GHZ}- and \W-functions defined in Example~\ref{ex:WandGHZ} we
arrive
at the unique polynomials \eqref{eqn:mlGHZ} and \eqref{eqn:mlW}.
\begin{equation}\label{eqn:mlGHZ}
f_{{\sf GHZ}}(x_1,x_2,x_3)=1 - x_1 - x_2 + x_1 x_2 - x_3 + x_1 x_3 + x_2 x_3
\end{equation}
\begin{equation}\label{eqn:mlW}
f_{\W}(x_1,x_2,x_3) = x_1 + x_2 + x_3 - 2 x_1 x_2 - 2 x_1 x_3 - 2 x_2 x_3 + 3 x_1
x_2 x_3
\end{equation}
These polynomials \eqref{eqn:mlGHZ} and \eqref{eqn:mlW} are readily translated into
categorical tensor networks.
\end{example}

\begin{example}[Network realisation of \W- and {\sf GHZ}-states]
A network realization of \W- and {\sf GHZ}-states in our framework then follows by
post-selecting the relevant network to $\ket{1}$ on the output bit --- leaving the
input qubits to
represent a \W- or {\sf GHZ}-state respectively. An example of this is shown in
Figure~\ref{fig:Wfunction}.
\end{example}

\begin{figure}[h]
\includegraphics[width=.90\textwidth]{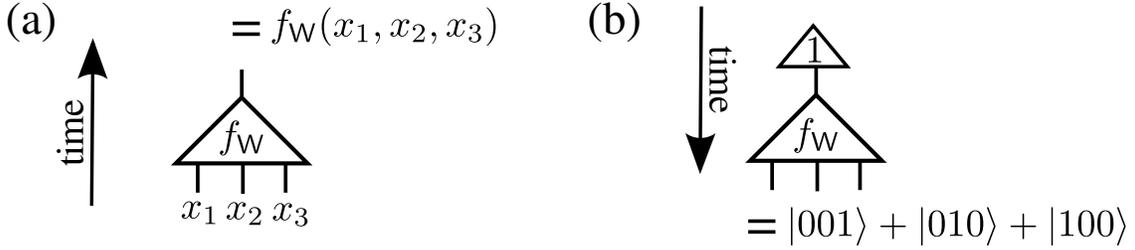}
\caption{Left (a) the circuit realization (internal to the triangle) of the function
$f_\W$ of e.g.\ \eqref{eqn:fWBoolean} which outputs logical-one given input $\ket{x_1
x_2 x_3}=$ $\ket{001}$, $\ket{010}$ and $\ket{100}$ and logical-zero
otherwise. Right (b) reversing time and setting the output to $\ket{1}$ (e.g.\
post-selection) gives a network representing the \W-state.  The na\"{i}ve
realization of $f_\W$ is given in Figure~\ref{fig:CategoricalWStates}
with an optimized co-algebraic construction shown in Figure
\ref{fig:CategoricalWStates}.}\label{fig:Wfunction}
\end{figure}

\begin{figure}[h]
\includegraphics[width=.50\textwidth]{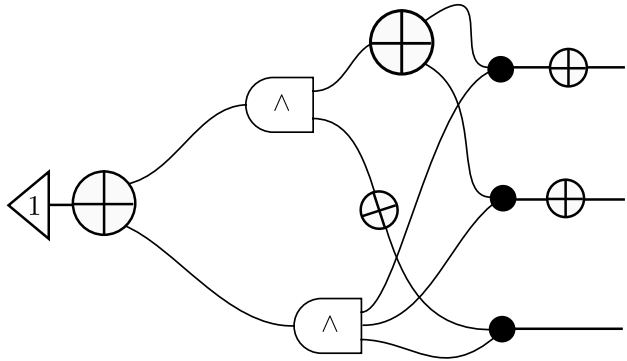}
\caption{Na\"{i}ve CTNS realization of the familiar \W-state
$\ket{001}+\ket{010}+\ket{100}$.  A standard (temporal) acyclic classical circuit
decomposition in terms of the \XOR-algebra realizes the function $f_\W$ of three
bits.  This function is given a representation on tensors.  As illustrated, the
networks input
is post selected to $\ket{1}$ to realize the desired
\W-state.}\label{fig:temporalWStates}
\end{figure}

\begin{figure}[h]
\includegraphics[width=.950\textwidth]{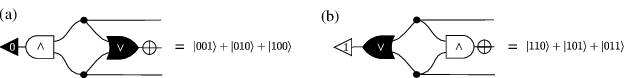}
\caption{\W-class states in the categorical tensor network state formalism.  (a)
is the standard \W-state.  (b) is found from applying De Morgan's law (see
Section~\ref{sec:AND}) to (a) and rearranging after inserting inverters on the output
legs. Notice the atemporal nature of the circuits, as
one gate is used forwards, and the other backwards.}\label{fig:CategoricalWStates}
\end{figure}

Two different categorical constructions for the building blocks of the \W-state are shown
in Figure~\ref{fig:temporalWStates} and  Figure~\ref{fig:CategoricalWStates}. Notice that in Figure~\ref{fig:CategoricalWStates} the
resulting tensor network forms an atemporal classical circuit and is much more efficient than
the na\"{i}ve construction in Figure~\ref{fig:temporalWStates}. Moreover by
appropriately daisy-chaining the networks in
Figure~\ref{fig:CategoricalWStates} we construct a categorical tensor network for
an $n$-party \W-state
as shown in Figure~\ref{fig:nPartyCategoricalWStates}. The resulting form of this tensor network is entirely equivalent
(up to regauging) to the MPS description given earlier, but now reveals internal structure of the state in terms of CTNS building blocks.

\begin{figure}[h]
\includegraphics[width=.950\textwidth]{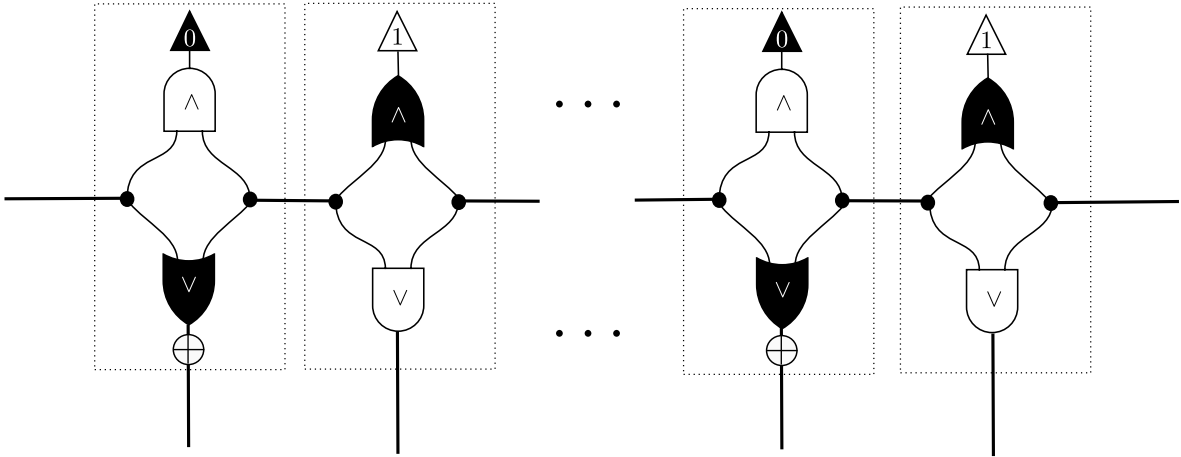}
\caption{\W-state ($n$-party) in the categorical tensor network state formalism.
The comb-like feature of efficient network contraction remains, with the
internal structure of the network components exposed in terms of well understood
algebraic structures.}\label{fig:nPartyCategoricalWStates}
\end{figure}

\subsection{Describing states with complex coefficients}
Boolean states, such as the {\sf GHZ}- and \W-states, are typified by being
superpositions of computational basis states with equal real coefficients (in both cases, these coefficients take only binary values, $0$ and $1$). 
In this section, we will permit a minor extension to binary superposition input/output states by considering arbitrary rank-1 tensors within our
otherwise Boolean tensor networks. This is
illustrated by a simple example:

\begin{example}[Network realization of
$\ket{\psi}=\ket{01}+\ket{10}+\alpha_k\ket{11}$]\label{ex:example1}
We will now design a network to realize the state
$\ket{01}+\ket{10}+\alpha_k\ket{11}$.  The first step is to write down a function
$f_S$
such that
\begin{equation}
 f_S(0,1) = f_S(1,0) = f_S(1,1) = 1
\end{equation}
and $f_S(00)=0$ (in the present case, $f_S$ is the logical \OR-gate). We post
select the network output on $\ket{1}$, which yields the state
$\ket{01}+\ket{10}+\ket{11}$, see Figure \ref{fig:example1}(a). The next step
is to realize a diagonal operator, that acts as identity on all inputs, except
$\ket{11}$ which gets sent to $\alpha_k\ket{11}$.  To do this, we design a function
$f_d$ such that
\begin{equation}
 f_d(0,1)=f_d(1,0)=f_d(0,0)=0
\end{equation}
and $f_d(1,1)=1$ (in the present case, $f_d$ is the logical \AND-gate).  This
diagonal, takes the form in Figure \ref{fig:example1}(b). The final state
$\ket{\psi}=\ket{01}+\ket{10}+\alpha_k\ket{11}$ is realized by connecting both
networks, leading
to Figure \ref{fig:example1}(c).
\end{example}

\begin{figure}[h]
\includegraphics[width=.950\textwidth]{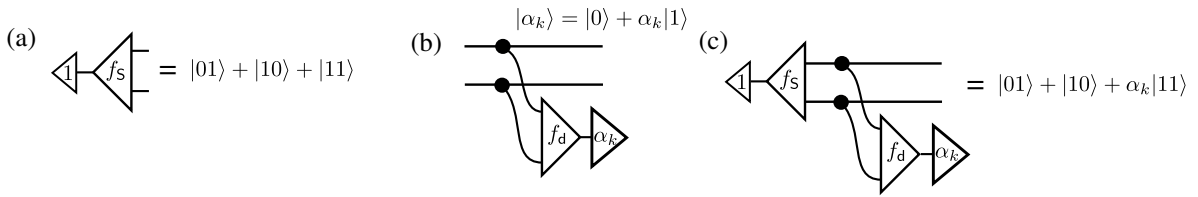}
\caption{Categorical tensor network representing
state $\ket{\psi}=\ket{01}+\ket{10}+\alpha_k\ket{11}$, as explained in Example \ref{ex:example1}.}\label{fig:example1}
\end{figure}

\section{Proof of The Main Theorems}\label{sec:mainresult}
We are now in a position to state the main theorem of this work. Specifically, we
have a constructive
method to realize any quantum state in terms of a categorical tensor
network. (A corollary of our exhaustive factorization of quantum states
into tensor networks is a new type of quantum network universality proof.  To avoid
confusion, we point out that past universality proofs in the gate model already
imply that the linear fragment (Figure~\ref{fig:F2-presentation}) together with local
gates is quantum universal.  However, the known universality results clearly do not
provide a method to factor a
state into a tensor network!  Indeed, the decomposition or factorization of a state
into a tensor network is an entirely different problem which we address here.) We
state and prove the theorem for the case of qubit. The higher dimensional case of
qudits follows from known results that
any $d$-state switching function can be expressed as a polynomial and realized as a
connected network \cite{MVL1, MVL2, MVL3}.  The theorem can be stated as

\begin{theorem}[Tensor network representation of quantum states]\label{theorem:rep}
For any state $\ket{\psi}$ of $n$-qubits with the form 
\begin{equation}\label{eqn:theorem:rep}
 \ket{\psi} = \sum_{j=1}^k \alpha_j\ket{\phi_j}, 
\end{equation}
where $\alpha_j$ are complex coefficients and for each $j$ the state $\ket{\phi_j}$
is an equal superposition
of a set of computational basis states, it can be represented as a network containing
tensors from the
quantum Boolean calculus (Figures \ref{fig:F2-presentation} and \ref{fig:extraF2}),
together with input/ouput states of the form
$\ket{\alpha_j}:=\ket{0}+\alpha_j\ket{1}$.
\end{theorem}

Notice that an arbitrary state can be brought into the form required of $\ket{\psi}$ by composing it as $k = 2^n$
terms with each state $\ket{\phi_j}$ being a single distinct computational basis
state. The proof is simplified by invoking some supporting lemmas.

\begin{lemma}
There exists
a map $g$ represented by a tensor network taking diagonal maps in
$\bigotimes_n\7C^2\rightarrow\bigotimes_n\7C^2$ onto quantum states in
$\bigotimes_n\7C^2$.
\end{lemma}
 
\begin{proof}
 Let $\2D$ be a diagonal map in
$\bigotimes_n\7C^2\rightarrow\bigotimes_n\7C^2$. We write
$\2D = \sum_{\textbf x\in\{0,1\}^n} \alpha_{\textbf x} \ket{\textbf x}\bra{\textbf
x}$ and
proceed as follows (where the term $\2D\circ \bigotimes_n(\ket{0} + \ket{1})$
immediately yields the desirable tensor network depiction)

\begin{align*}
g\{\2D\} &:= \2D\circ \bigotimes_n(\ket{0} + \ket{1}) = \2D \circ
\sum_{\textbf{y}\in\{0,1\}^n} \ket{\textbf y} =
\sum_{\textbf{x},\textbf{y}\in\{0,1\}^n}\alpha_{\textbf x}\ket{\textbf
x}\braket{\textbf x}{\textbf y} = \\
&= \sum_{\textbf
x\in\{0,1\}^n}\delta_{\textbf{xy}}\alpha_{\textbf x}\ket{\textbf x}=\sum_{\textbf
x\in\{0,1\}^n}\alpha_{\textbf x}
\ket{\textbf x} 
\end{align*}

\end{proof}

\begin{lemma}\label{lemma:h}
There exists a map $h$ represented by a tensor network taking quantum
states in $\bigotimes_n\7C^2$ onto diagonal maps in
$\bigotimes_n\7C^2\rightarrow\bigotimes_n\7C^2$. 
\end{lemma}

\begin{proof}
Let $\ket{\psi}$ be a quantum state in
$\bigotimes_n\7C^2$.   Let $\2D$ be a diagonal map in
$\bigotimes_n\7C^2\rightarrow\bigotimes_n\7C^2$. We write
$\2D = \sum_{\textbf x\in\{0,1\}^n} \alpha_{\textbf x} \ket{\textbf x}\bra{\textbf
x}$ and
proceed as follows (where the term $\bigotimes_n (\sum_{i=0/1}\ket{ii}\bra{i})\circ
\sum_{\textbf x\in\{0,1\}^n} \alpha_{\textbf x} \ket{\textbf x}$
immediately yields the desirable tensor network depiction)
\begin{align*}
h'\{\ket{\psi}\} &:= \bigotimes_n (\sum_{i=0/1}\ket{ii}\bra{i})\circ \sum_{\textbf
x\in\{0,1\}^n} \alpha_{\textbf x}
\ket{\textbf x} = \sum_{\textbf y\in\{0,1\}^n}
\ket{\textbf{yy}}\bra{\textbf y} \circ \sum_{\textbf
x\in\{0,1\}^n} \alpha_{\textbf x} \ket{\textbf x} =\\
&= \sum_{\textbf{x},\textbf{y}\in\{0,1\}^n}\ket{\textbf{yy}}\alpha_{\textbf
x}\delta_{\textbf{xy}} = \sum_{\textbf
x\in\{0,1\}^n}\alpha_x\ket{\textbf{xx}} 
\end{align*}

and then we now write $h$ in terms of $h'$ 
\begin{align*}
h\{\ket{\psi}\} &:= \bigotimes_n(\sum_{i=0/1}\bra{ii})\circ h'\{\2D\} =
\sum_{\textbf y\in\{0,1\}^n}\bra{\textbf{yy}}\circ \sum_{\textbf
x\in\{0,1\}^n}\alpha_{\textbf x}\ket{\textbf{xx}} =\\
&= \sum_{\textbf{y, x}\in\{0,1\}^n}\alpha_{\textbf x} \ket{\textbf
x}\bra{\textbf y}\delta_{\textbf{yx}} = \sum_{\textbf x\in\{0,1\}^n}\alpha_{\textbf
x}\ket{\textbf x}\bra{\textbf x} 
\end{align*}

\end{proof}

\begin{corollary}
 It follows that $g\{h\{\ket{\psi}\}\} = \eye_\psi\circ\ket{\psi}=\ket{\psi}$ and
$h\{g\{\2D\}\} =
\eye_{\2D}\circ \2D= \2D$ and hence we have inverses for $g$ and $h$ establishing an
isomorphism between diagonal operators in 
$\bigotimes_n\7C^2\rightarrow\bigotimes_n\7C^2$ and quantum states in
$\bigotimes_n\7C^2$.  
\end{corollary}

With the supporting lemmas in place, we will now proceed to prove
Theorem~\ref{theorem:rep}. 
\begin{proof}
Returning to our particular expression for an arbitrary quantum state (note that
in Equation \eqref{eqn:ptr} we now append an extra term, by letting the sum run from
$j=0$
instead of $j=1$, the use of this will become clear below)
\begin{equation}\label{eqn:ptr}
\ket{\psi} = \sum_{j=0}^k \alpha_j\ket{\phi_j},
\end{equation}
where $\alpha_j$ are complex coefficients and for each $j$ the
state $\ket{\phi_j}$ is an equal superposition of a set of computational basis
states, we will explain how $k+1$ asynchronous circuits~\cite{Weg87} are
used to factor the state, and express the state as a CTNS (here and in what
follows $k$ is the highest term in the sum from Equation \eqref{eqn:theorem:rep}).   

We proceed by returning to our original expression \eqref{eqn:theorem:rep} from
Theorem \ref{theorem:rep} (starting from $j=1$) with the coefficients removed
\begin{equation}\label{eqn:coe}
\ket{\psi'}=\sum_{j=1}^k \ket{\phi_j}
\end{equation}
Each individual term in Equation \eqref{eqn:coe} is then expressed in the
computational basis and used to form a set denoted $\2L^+$. All corresponding bit
patterns of the same dimension not appearing in this expression form a second set
$\2L^-$ (where clearly $\2L^+\cap\2L^-=\emptyset$ and $\2L^+\cup\2L^-=\{0,1\}^n$,
where $n$ is the number of qubits in the desired state).  We proceed to construct
a function $f_0$ that outputs logical-one on all 
input bit strings in $\2L^+$ and outputs local zero on all input bit strings in
$\2L^-$.  The function acts on $n+1$ bits, the inputs are given on the right of the
tensor symbol and the output on the left of the tensor symbol in \eqref{eqn:func} 
\begin{equation}\label{eqn:func}
\ket{\psi_f} = \sum_{\textbf x\in\{0,1\}^n} \ket{\textbf x}\otimes\ket{f(\textbf x)}
\end{equation} 
where $f(\textbf{x}):\7B^n\rightarrow \7B::\textbf{x}\mapsto f(\textbf{x})$ was given
the representation on
quantum states in Section \ref{sec:CTN}.  Post selecting the networks output (the
rightmost bit in Equation \eqref{eqn:func}) $\ket{1}$ realizes the desired
superposition of terms, with all coefficients of the terms and hence relative phases
equal. 
\be
\ket{\psi_f'} = \sum_{\textbf x\in\{0,1\}^n} \braket{1}{f(\textbf x)}\ket{\textbf x}
\ee

For our specific construction, depicted in Figure~\ref{fig:fullCTNS}, we proceed by
inverting the output of the function (e.g.\ $f_0 \mapsto f_0\oplus 1$).  We then post
select the output of the function to the state $\ket{\alpha_0}=\ket{0} +
\alpha_0\ket{1}=\ket{0}$ for the choice $\alpha_0=0$.  This handles the $j=0$ term
in the sum Equation \eqref{eqn:ptr}.  

To adjust the amplitudes of the desired state from Equation \eqref{eqn:ptr}, we will
construct tensors that represent diagonal operators. For the $j$th term in
$\ket{\psi}$ with coefficient $\alpha_j$,
we again construct a function $f_j$.  We represent $\ket{\phi_j}$ in the
computational basis, and each term in this expression is used to form a set denoted
$\2L^+$. All corresponding bit patterns of the same dimension not appearing in this
expression form a second set $\2L^-$.  We proceed to construct $f_j$ to output
logical-one on all input bit strings in $\2L^+$ and outputs logical-zero on all input
bit strings in $\2L^-$.  The network is then post selected to $\ket{0} +
\alpha_j\ket{1}$ which results in states of the form 
\be 
\ket{\psi_D} =  \sum_{\textbf x\in\{0,1\}^n} \braket{0}{f(\textbf x)}\ket{\textbf x}
+ \alpha_j \sum_{\textbf x\in\{0,1\}^n}
\braket{1}{f(\textbf x)}\ket{\textbf x}
\ee
and we transform $f_j$ into a diagonal operator having entries $\in\{1,\alpha_j\}$ by
applying the map $h$ from Lemma
\ref{lemma:h} resulting in the diagonal map 
\be 
\2D_j =  \sum_{\textbf x\in\{0,1\}^n} \braket{0}{f(\textbf x)}\ket{\textbf
x}\bra{\textbf x} + \alpha_j \sum_{\textbf x\in\{0,1\}^n}
\braket{1}{f(\textbf x)}\ket{\textbf x}\bra{\textbf x}
\ee 
We will apply $k$ such commuting maps $\2D_j$ to the initial state, accounting for
$k$
asynchronous circuits.  The operators are composed by means of $n$
co-{\sf COPY}-dots from Section~\ref{sec:co-diagonal} (see Figure \ref{fig:example1}
and Example \ref{ex:example1}). There will be a single output with open legs which
gives the state.  Each of the $n$ {\sf COPY}-dots will then require $k+2$ legs.  The
resulting construction then gives tensor networks with the form shown in
Figure~\ref{fig:fullCTNS}.
\end{proof}

\begin{figure}[h]
\includegraphics[width=.750\textwidth]{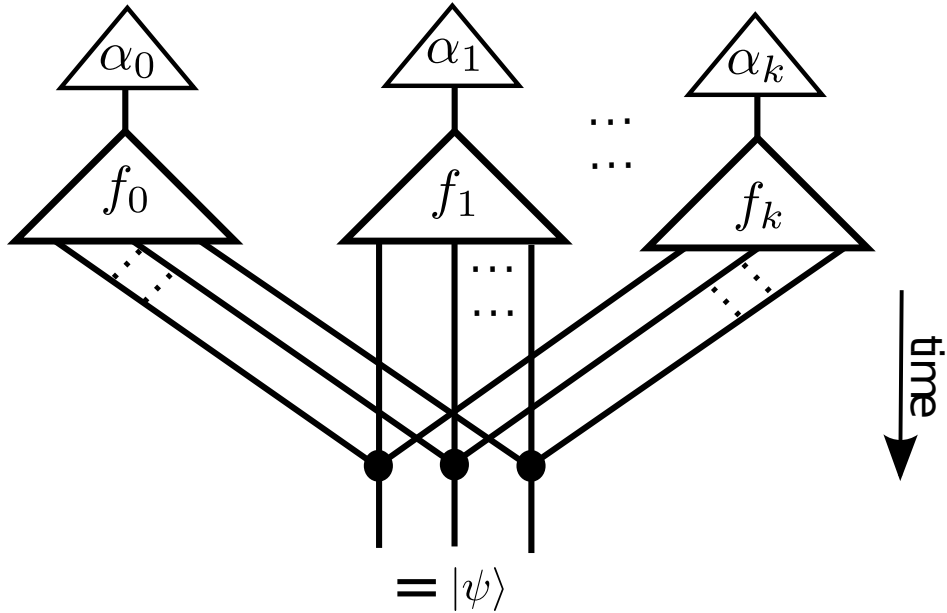}
\caption{The CTNS for a state $\ket{\psi}$ resulting from our exhaustive construction procedure.}\label{fig:fullCTNS}
\end{figure}

\begin{remark}[Qudit states]\label{re:nonqubits} In Theorem
\ref{theorem:rep} we considered the arbitrary states of $n$-qubits. By using
multi-valued logic (also called $d$-state switching, or many-valued logic),
it is possible to define a universal gate set similar to what was done for the case of qubits and
so equivalently construct a CTNS for $n$-body qudit systems \cite{MVL1, MVL2, MVL3}.
\end{remark}

\begin{definition}[Generalized Polynomial Boolean States
(GPBS)]\label{def:GPBS}
 The size of a tensor network is the number of tensors it contains and its depth is
the maximal length of a path from any tensor to any other. Consider families
of uniform circuits, built from the {\sf XOR}-algebra, that is the bounded fan-in
gates {\sf AND}, {\sf XOR}, {\sf COPY}, of arity two and the constant $\ket{1}$.  We
will index
these circuit families by bounding the circuit depth, which also has the impact of
bounding the maximum fan-in and the circuit size.  We will then consider categorical
tensor networks to represent states $\ket{\psi}$ of $n$ interacting $d$-level
systems constructed from bounded circuit families to realize each $f_j$ with the form
given in Figure \ref{fig:fullCTNS}.  We will proceed by indexing these families of
categorical tensor networks in terms of
$k$, the maximum depth of any given circuit realizing any function $f_j$ in the
network. We will then bound the number of such functions $f_j$ to be
at most some polynomial in $k$. We then determine how $k(n)$ changes.  This works by
considering circuit families and categorical tensor networks of a given form, used to
represent quantum states on increasingly many subsystems $n$.  If $k(n)$ is bounded
by a polynomial in $n$ the categorical tensor network has an efficient description. 
We index such families as {\sf C}($k$), and refer to them as Generalized Polynomial
Boolean States (GPBS). 
\end{definition}

\begin{theorem}\label{corollary:states}
{\upshape GPBS} from Definition \ref{def:GPBS} are sampled exactly in the
computational basis in time and space complexity bounded by {\upshape poly($k$)}.
\end{theorem}

\begin{proof}
To prove Theorem \ref{corollary:states} we begin first by considering a 
qudit state vector $\ket{x_0,x_1,...,x_n}$ for specific
$x_0,x_1,...,x_n\in\{0,1,...,d-1\}$.  We wish to know the coefficient
$\braket{x_0,x_1,...,x_n}{{\sf C}}$, where {\sf C} is a CTNS representing a GPBS. 
The {\sf COPY}-gates in the construction from Figure \ref{fig:fullCTNS} map 
\be
\ket{x_0,x_1,...,x_n}\mapsto \bigotimes_{\text{poly}(k)} \ket{x_0,x_1,...,x_n}
\ee 
each of these poly($k$) vectors will be acted on by a network realizing $f_j$ post
selected to the state $\ket{\alpha_j}$.  Hence, to sample the network {\sf C}
amounts
evaluating the sum 
\be
\braket{x_0,x_1,...,x_n}{{\sf C}}=\sum_{j\in \text{poly}(k)}
\braket{f_j(x_0,x_1,...,x_n)}{\alpha_j} = c\in \7C
\ee
The proof then follows by simply evaluating each of the poly($k$) poly($k$)-depth 
functions $f_j(x_0,x_1,...,x_n)$ and then summing the inner products
$\braket{f_j(x_0,x_1,...,x_n)}{\alpha_j}$. 
\end{proof}

We note that the construction in Theorem~\ref{theorem:rep} automatically groups basis
states with the same coefficients
$\alpha_j$, of the $k$ terms.  Further reductions are also possible if say a
given set of
coefficients are given by products of other coefficients. While this construction
does prove the existence of a CTNS (along with how to build it) our
construction will not render efficient representations for
general cases, as one might expect. Indeed, there is no guarantee that any of the
$k+1$ switching functions are efficient in their complexity, nor that the
resulting
complete network is contractible. The latter property is in fact confounded by the
presence of fan-in (up to $k+2$ legs) of the $n$ co-{\sf COPY}-dots (the presence
also
implies that the network
cannot represent a deterministic physical preparation procedure~\cite{BB11}). However,
as we saw earlier with string-bond states in Figure~\ref{fig:stringbonds}, the {\sf
COPY}~dot breaks up 
when computational basis states are inputted. For our general decomposition in in Figure~\ref{fig:fullCTNS},
this case causes the $k+1$ Boolean switching functions to similarly decouple.
Intuitively, if we
further restrict ourselves
to $k+1$ being polynomial in $n$, and additionally that each switching function has a depth which is also polynomial
in $n$, then the subsequent evaluation of the amplitude of the state is efficient for
any computational basis state (see Definition
\ref{def:GPBS} and Theorem \ref{corollary:states}).  This is a weak requirement in
practice and interestingly, does not depend on the internal geometry of the networks
representing the functions, but only on their depth and size.  Thus we 
have found a new general class of states which can be sampled exactly and
efficiently. Finally
the construction was based on using acyclic-temporal
Boolean circuits. However, we have already seen that in the tensor context
wires can be bent around: it is not necessary for a tensor network to correspond
to a valid classical circuit. As the \W-state example (see
Figures~\ref{fig:CategoricalWStates} and \ref{fig:nPartyCategoricalWStates})
illustrated the tensor networks
can be much simpler once this new freedom is exploited.

\section{Outlook and Concluding Remarks}\label{sec:conclusion}

We have introduced methods from category theory and algebra to tensor network states.  Our main
focus has been on tensor networks built from contracting Boolean tensors and we
have outlined their algebraic properties. The logical conclusion of
this approach has led us to a exhaustive CTNS decomposition of
quantum states. From this we obtained a class of quantum states, which
can
be sampled in the computational basis efficiently and exactly.  The expressiveness 
and power of this new method was further illustrated by
considering several simple test cases: we considered internal structure of some
MPS states, e.g.\ {\sf GHZ}- and \W-states. We have opened up
some future potential research directions.  In particular, beyond the form of our
construction there is an open question as to whether the algebraic properties 
of some subset of the tensors in CTNS can enable efficiently contractible
networks beyond those already known which are based on topology (like MPS) and
additional
unitary/isometric constraints (MERA). In this way future studies of CTNS may lead to new
classes of states and algorithms which will help challenge and shape our
understanding of many-body physics.  


\begin{acknowledgments}
We thank John Baez, Tomi Johnson, Martin Plenio, Vlatko
Vedral, Mike Shulman, Ville Bergholm, Samson Abramsky, Bob Coecke, Chris Heunen and
Guifr\'{e} Vidal. JDB
received support from the EPSRC and completed large parts of
this work visiting the Centre for Quantum Technologies, at
the National University of Singapore (these visits were hosted by Vlatko
Vedral).  SRC and DJ thank the National Research Foundation and
the Ministry of Education of Singapore for support. DJ acknowledges
support from the ESF program EuroQUAM (EPSRC
grant EP/E041612/1), the EPSRC (UK) through the
QIP IRC (GR/S82176/01), and the European Commission under
the Marie Curie programme through QIPEST.  Earlier versions of this work circulated
for just over one year prior to being uploaded to the arXiv, the most widespread
appearing under the title \textit{Algebra and co-algebra on Categorical Tensor Network
States}.   
\end{acknowledgments}

\appendix

\section{Algebra on quantum states}\label{sec:newalgebra}
We are concerned with a network theory of quantum states.  This on the one hand can
be used as a tool to solve problems about states and operators in quantum theory, but
does have a physical interpretation on the other.  This is not foundational
\textit{per se} but instead largely based on what one 
might call an operational interpretation of quantum states and processes.  We call
an algebra a pairing on a vector space, taking two vectors and producing a third
(you might instead call it a monoid if there is a unit, and then a group if the
set of considered vectors is closed under the product).  Let's now examine how every
tripartite quantum state forms an algebra~\cite{course}.

Consider a tripartite quantum state (subsystems labeled 1,2 and 3), and then ask the
question: ``how would the state of the third system change after measurement of 
systems one and two?''   Enter Algebras: as stated, an algebra on a vector space, or
on a Hilbert space is formed by a \textit{product} taking two elements from the
vector space to produce a third element in the vector space.  Algebra on states can
then be studied by considering duality of the state, that is considering the
adjunction between the maps of type
\begin{equation}
\eye\rightarrow \2 H\otimes \2H \otimes \2H~~~~~~~~\text{and}~~~~~~~~\overline{\2
H}\otimes \overline{\2H} \rightarrow \2H
\end{equation}
This duality is made evident by using the $\dagger$-compact structure of the category
(e.g.\ the cups and caps).  It is given vivid physical meaning by considering the
effect measuring (that is two events) two components of a state has on the third
component.

\begin{remark}[Overbar notation on Spaces] Given a Hilbert space $\2 H$, we can
consider the Hilbert space $\overline{\2H}$ which can be simply thought of as the
Hilbert space $\2 H$ will all basis vectors complex conjugates (overbar).  That is,
$\overline{\2H}$ is a vector space whose elements are in one-to-one
correspondence with the elements of $\2 H$:
\begin{equation}
    \overline{\2H} = \{\overline{v} \mid v \in \2 H\},
\end{equation}
with the following rules for addition and scalar multiplication:
\begin{equation}
    \overline v + \overline{w} =
\overline{\,v+w\,}\quad\text{and}\quad\alpha\,\overline v = \overline{\,\overline
\alpha \,v\,}.
\end{equation}
\end{remark}

\begin{remark}[Definition of Algebra]
We consider an algebra as a vector space $\2A$ endowed with a product, taking a pair
of elements (e.g.\ from $\2A\otimes \2A$) and producing an element in $\2A$.  So the
product is a map $\2A\otimes \2A\rightarrow \2A$, which may not be associative or
have a unit (that is, a multiplicative identity --- see Example \ref{ex:weakunits}
for an
example of an algebra on a quantum state without a unit).
\end{remark}

\begin{observation}[Every tripartite Quantum State Forms an Algebra]
Let $\ket{\psi}\in \2 H\otimes \2H \otimes \2H$ be a quantum state and let
$M_i$, $M_j$ be complete sets of measurement operators.  Then $(\ket{\psi}, M_i,M_j)$
forms
an algebra.
\begin{center}
\includegraphics[width=.750\textwidth]{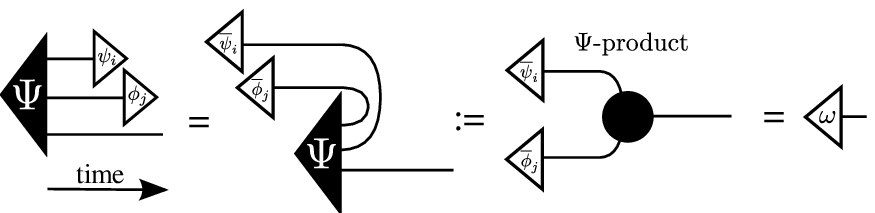}
\end{center}
\end{observation}

The quantum state $\ket{\Psi} = \sum_{ijk}\psi^{ijk}\ket{ijk}$ is drawn as a
triangle, with
the identity operator on each
subsystem acting as time goes to the right on the page (represented as a wire). 
Projective measurements
with respect to $M_i$ and $M_j$ are made.  We define these complete
measurement operators as
\begin{gather}
M_1=\sum_{i=1}^{N}i\ket{\psi_i}\bra{\psi_i}
\end{gather}
\begin{gather}
M_2=\sum_{j=1}^{N}j\ket{\phi_j}\bra{\phi_j}
\end{gather}
such that we recover the identity operator on the $N$-level subsystem viz
\begin{equation}
\sum_{j=1}^{N}\ket{\phi_j}\bra{\phi_j}=\sum_{i=1}^{N}\ket{\psi_i}\bra{\psi_i}=\eye_N
\end{equation}
The measurements result in eigenvalues $i,j$ leaving the state of the unmeasured
system in
\begin{equation}
\ket{\omega} = \sum_{xyz}\psi^{xyz} \braket{\overline{\psi}^x}{x}
\braket{\overline{\phi}^y}{y} \ket{z}
\end{equation}
where $\bra{\overline{Q}} \bydef \ket{Q}^\top$ that is, the transpose is factored
into: (i) taking the dagger (diagrammatically this mirrors states across the page)
and (ii) taking the complex conjugate.  Hence,
\begin{equation}
\ket{\overline{Q}}^\dagger = \ket{Q}^\top = \bra{\overline{Q}} =
\overline{\ket{Q}^\dagger}
\end{equation}
and if we pick a real valued basis for $\ket{x},\ket{y},\ket{z}=\ket{0},\ket{1}$ we
recover
\begin{equation}
\ket{\omega} = \sum_{xyz}\psi^{xyz}
\braket{x}{\psi_x}\braket{y}{\phi_y}\ket{z}
\end{equation}

As stated, this physical
interpretation is not our main interest. Even in its absence, we're able to write
down and represent a quantum
state purely in terms of a connected network, where each component is fully defined
in terms of algebraic laws.

\section{\XOR-algebra}\label{sec:Boolean}
Here we review the concept of an algebraic normal form (\ANF) for Boolean
polynomials, commonly known as \PPRM s.  See the reference
book~\cite{Davio78} and the historical references~\cite{Cohn62,
xor70} for further details. 

\begin{definition}
The \XOR{}-algebra forms a commutative ring with presentation
$M=\{\7B,\wedge,\oplus\}$ where the following product is called \XOR
\begin{equation}
\text{---}\oplus\text{---}:\7B\times \7B\mapsto \7B:: (a,b)\rightarrow
a+b-ab~\text{mod}~2
\end{equation}
and conjunction is given as
\begin{equation}
\text{---}\wedge\text{---}:\7B\times \7B\mapsto \7B:: (a,b)\rightarrow a\cdot b, 
\end{equation}
where $a\cdot b$ is regular multiplication over the reals.  One defines left negation
$\neg(\text{---})$ in terms of $\oplus$ as
$\neg(\text{---})\equiv$
\begin{equation}
  \text{1}\oplus(\text{---}):\7B \mapsto \7B:: a\rightarrow 1-a.
\end{equation}
In the \XOR-algebra, 1-5 hold.  (i) $a\oplus 0 = a$, (ii) $a\oplus 1 = \neg a$, (iii)
$a\oplus a = 0$, (iv) $a\oplus \neg a = 1$ and (v) $a\vee b = a\oplus b\oplus
(a\wedge b)$.  Hence, $0$ is the unit of \XOR{} and $1$ is the unit of \AND. The 5th
rule
reduces to $a\vee b = a\oplus b$ whenever $a\wedge b=0$, which is the case for
disjoint ($\text{mod}~2$) sums.  The 
truth table for $\AND$ follows
\begin{center}
\begin{tabular}{c|c|c}
$~x_1~$ & $~x_2~$ & $f(x_1,x_2)=x_1\wedge x_2$ \\ \hline
0 & 0 & 0 \\
0 & 1 & 0 \\
1 & 0 & 0 \\
1 & 1 & 1
\end{tabular}
\end{center}

\end{definition}

\begin{definition}\label{def:FPRM}
Any Boolean equation may be uniquely expanded to the fixed polarity
Reed-Muller form as:\setlength{\arraycolsep}{0.140em}
\begin{eqnarray}\label{eqn:rm_exp}
&&f(x_1,x_2,...,x_k) = c_0\oplus c_1 x_1^{\sigma_1}\oplus c_2 x_2^{\sigma_2}\oplus
\cdots \oplus c_n x_n^{\sigma_n}\oplus\nonumber\\
&&~~~~~~~~c_{n+1}x_1^{\sigma_1} x_n^{\sigma_{n}}\oplus \cdots \oplus
c_{2k-1}x_1^{\sigma_1} x_2^{\sigma_2},...,x_k^{\sigma_k},
\end{eqnarray}
where selection variable $\sigma_i\in \{0,1\}$, literal
$x_i^{\sigma_i}$ represents a variable or its negation and any $c$
term labeled $c_0$ through $c_j$ is a binary constant $0$ or $1$. In
Equation~\eqref{eqn:rm_exp} only fixed polarity variables appear such that
each is in either un-complemented or complemented form.
\end{definition}

Let us now consider derivation of the form from Definition~\ref{def:FPRM}.  Because
of the structure of
the algebra, without loss of generality, one avoids keeping track of
indices in the $N$ node case, by considering the case where $N\equiv 2^n=8$.

\begin{example}

The vector $\underline{c}=(c_0,c_1,c_2,c_3,c_4,c_5,c_6,c_7,)^\intercal$ represents
all possible outputs of any function $f(x_1,x_2,x_3)$ over the
algebra formed from linear extension of $\7Z_2\times \7Z_2\times \7Z_2$.  We wish to
construct a normal form in terms of the vector $\underline{c}$, where each $c_i\in
\{0,1\}$, and therefore $\underline{c}$ is a selection vector
that simply represents the output of the function $f:\7B\times\7B\times\7B\rightarrow
\7B::(x_1,x_2,x_3)\mapsto f(x_1,x_2,x_3)$. One may expand $f$ as:

\begin{eqnarray}\label{eqn:generic}
f(x_1,x_2,x_3) &=& (c_0\cdot \neg x_1\cdot \neg x_2\cdot \neg
x_3)\vee(c_1\cdot \neg x_1\cdot \neg x_2\cdot x_3)\vee(c_2\cdot
\neg x_1\cdot x_2\cdot \neg x_3)\nonumber\\
&&\vee(c_3\cdot \neg x_1\cdot x_2\cdot x_3)\vee(c_4\cdot x_1\cdot
\neg x_2\cdot \neg x_3)\vee(c_5\cdot x_1\cdot
\neg x_2\cdot x_3)\nonumber\\
&&\vee(c_6\cdot x_1\cdot x_2\cdot \neg x_3)\vee(c_7\cdot x_1\cdot
x_2\cdot x_3)
\end{eqnarray}

Since each disjunctive term is disjoint the logical \OR{} operation may
be replaced with the logical \XOR{} operation.  By making the substitution $ \neg
a=a\oplus 1$ for all variables and rearranging terms one arrives at
the following normal form.  (For instance, $\neg x_1\cdot
\neg x_2\cdot\neg x_3=(1\oplus x_1)\cdot(1\oplus x_2)\cdot(1\oplus
x_3)=(1\oplus x_1\oplus x_2\oplus x_2\cdot x_3)\cdot(1\oplus x_3)=
1\oplus x_1\oplus x_2\oplus x_3\oplus x_1\cdot x_3\oplus x_2\cdot
x_3\oplus x_1\cdot x_2\cdot x_3$.) 

\begin{eqnarray}\label{eqn:generic2}
f(x_1,x_2,x_3) &=&c_0\oplus(c_0\oplus c_4)\cdot x_1\oplus(c_0\oplus
c_2)\cdot x_2\oplus(c_0\oplus c_1)\cdot x_3\oplus(c_0\oplus
c_2\oplus c_4\oplus c_6)\cdot x_1\cdot
x_2\nonumber\\
&& \oplus(c_0\oplus c_1\oplus c_4 \oplus c_5)\cdot x_1\cdot
x_3\oplus(c_0 \oplus c_1 \oplus c_2 \oplus c_3)\cdot x_2\cdot x_3\nonumber\\
&&\oplus (c_0\oplus c_1\oplus c_2\oplus c_3\oplus c_4 \oplus c_5
\oplus c_6\oplus c_7)\cdot x_1\cdot x_2\cdot x_3
\end{eqnarray}

The set of linearly independent vectors, $\{x_1,x_2,x_3,x_1\cdot
x_2,x_1\cdot x_3,x_2\cdot x_3,x_1\cdot x_2\cdot x_3\}$ combined with
a set of scalars from Equation~\ref{eqn:generic2} spans the eight
dimensional space of the Hypercube representing the Algebra.
A similar form holds for arbitrary $N$.

\begin{eqnarray}\label{eqn:generic3}
f(x_1,x_2,x_3) &=& (a_1)\cdot x_1\oplus(a_2)\cdot
x_2\oplus(x_3)\cdot x_3\oplus(a_1\oplus a_2\oplus a_1\oplus
c_2)\cdot x_1\cdot
x_2\nonumber\\
&& \oplus(a_1\oplus a_3\oplus a_1\oplus c_3)\cdot x_1\cdot
x_3\oplus(a_2\oplus a_3\oplus a_2\oplus
c_3)\cdot x_2\cdot x_3\nonumber\\
&&\oplus (a_1\oplus a_2\oplus a_3\oplus a_1\oplus a_2\oplus
a_3)\cdot x_1\cdot x_2\cdot x_3
\end{eqnarray}
\end{example}



\end{document}